\documentclass[a4paper,UKenglish,cleveref, autoref, thm-restate]{lipics-v2021}
\pdfoutput=1 %uncomment to ensure pdflatex processing (mandatatory e.g. to submit to arXiv)
\hideLIPIcs  %uncomment to remove references to LIPIcs series (logo, DOI, ...), e.g. when preparing a pre-final version to be uploaded to arXiv or another public repository
\nolinenumbers

\newtheorem{obs}[theorem]{Observation}
\newcommand{\msize}[1]{{\,\left|#1\right|\,}}

\newcommand{\floor}[1]{\left\lfloor#1\right\rfloor}
\newcommand{\calS}{{\cal S}}

\title{On the Computational Complexity of Generalized Common Shape Puzzles} %TODO Please add

\author{Mutsunori Banbara}%
{Nagoya University, Nagoya, Japan}{}{https://orcid.org/0000-0002-5388-727X}{}

\author{Shin-ichi Minato}%
{Kyoto University, Kyoto, Japan}{}{}{}

\author{Hirotaka Ono}%
{Nagoya University, Nagoya, Japan}{}{https://orcid.org/0000-0003-0845-3947}{}

\author{Ryuhei Uehara\footnote{Corresponding author}}%
{School of Information Science, Japan Advanced Institute of Science and Technology, Japan}%
{uehara@jaist.ac.jp}{https://orcid.org/0000-0003-0895-3765}%
{}

\authorrunning{M. Banbara et al.}

\Copyright{Mutsunori Banbara, Shin-ichi Minato, Hirotaka Ono, and Ryuhei Uehara} 
\ccsdesc[100]%{\textcolor{red}{Replace ccsdesc macro with valid one}} %TODO mandatory: Please choose ACM 2012 classifications from https://dl.acm.org/ccs/ccs_flat.cfm 
{Theory of computation $\rightarrow$ Computational complexity and cryptography $\rightarrow$ Complexity classes,
Theory of computation $\rightarrow$ Computational complexity and cryptography $\rightarrow$ Problems, reducitons and completeness}

 %TODO mandatory; please add comma-separated list of keywords

\category{} %optional, e.g. invited paper

\relatedversion{} %optional, e.g. full version hosted on arXiv, HAL, or other respository/website
\EventEditors{}
\EventNoEds{1}
\EventLongTitle{}
\EventShortTitle{}
\EventAcronym{}
\EventYear{}
\EventDate{}
\EventLocation{}
\EventLogo{}
\SeriesVolume{}
\ArticleNo{}
\begin{document}

\maketitle

\begin{abstract}
In this study, we investigate the computational complexity of some variants of generalized puzzles.
We are provided with two sets $\calS_1$ and $\calS_2$ of polyominoes.
The first puzzle asks us to form the same shape using polyominoes in $\calS_1$ and $\calS_2$.
We demonstrate that this is polynomial-time solvable if $\calS_1$ and $\calS_2$ have constant numbers of polyominoes,
and it is strongly NP-complete in general.
The second puzzle allows us to make copies of the pieces in $\calS_1$ and $\calS_2$.
That is, a polyomino in $\calS_1$ can be used multiple times to form a shape.
This is a generalized version of the classical puzzle known as the common multiple shape puzzle.
For two polyominoes $P$ and $Q$, the common multiple shape is a shape that can be formed by 
many copies of $P$ and many copies of $Q$.
We show that the second puzzle is undecidable in general.
The undecidability is demonstrated by a reduction from a new type of undecidable puzzle based on tiling.
Nevertheless, certain concrete instances of the common multiple shape can be solved in a practical time.
We present a method for determining the common multiple shape for provided tuples of polyominoes
and outline concrete results, which improve on the previously known results in puzzle society.
\end{abstract}

\acknowledgements{I want to thank \dots}%optional

%\nolinenumbers %uncomment to disable line numbering

\section{Introduction}
\label{sec:intro}
Research on the computational complexity of puzzles and games has become increasingly important in theoretical computer science 
(see \cite{Uehara2023a} for a comprehensive survey).
Since the 1990s, numerous puzzles have been demonstrated to be NP-complete in general. 
These results provide a certain amount of some common sense regarding the NP class.
However, it has not been possible to capture certain puzzles, among which the sliding block puzzle is representative.
This has remained an open problem since Martin Gardner pointed out in the 1960s that 
a certain theorem is required to understand such puzzles.
However, after 40 years, Hearn and Demaine proposed a framework known as
constraint logic, and demonstrated that these puzzles are PSPACE-complete \cite{HearnDemaine2005,HearnDemaine2009}.
Combinatorial reconfiguration problems 
have been investigated towards an understanding of the PSPACE class \cite{ItoDemaineHarveyPapadimitriouSideriUeharaUno2011}.

With the developments in theoretical computer science in the past decade, 
new series of puzzles have been developed in puzzle society. 
In comparison to classical packing puzzles, one major property of these puzzles is that the goal is not explicitly stated.
The first example is the \emph{symmetric shape puzzle}. This puzzle asks us to form a symmetric shape
using a given set of pieces. It is extremely challenging to solve such a puzzle because 
we cannot be sure whether or not we are approaching the goal.
This property makes the puzzle very difficult, and in fact, 
only a few pieces are sufficient to cause this difficulty \cite{DKKMORRUU2020}. 
The second example is the \emph{anti-slide puzzle}. This puzzle asks us to interlock a given set of pieces.
A typical instance asks us to pack the given pieces into a frame so that any piece cannot be slide in the frame.
This puzzle is also difficult because the goal is not explicitly stated.
The computational complexity of this puzzle was recently investigated by \cite{MUH2023}.

\begin{figure}[thb]\centering
\includegraphics[bb=0 0 842 595,width=0.8\textwidth]{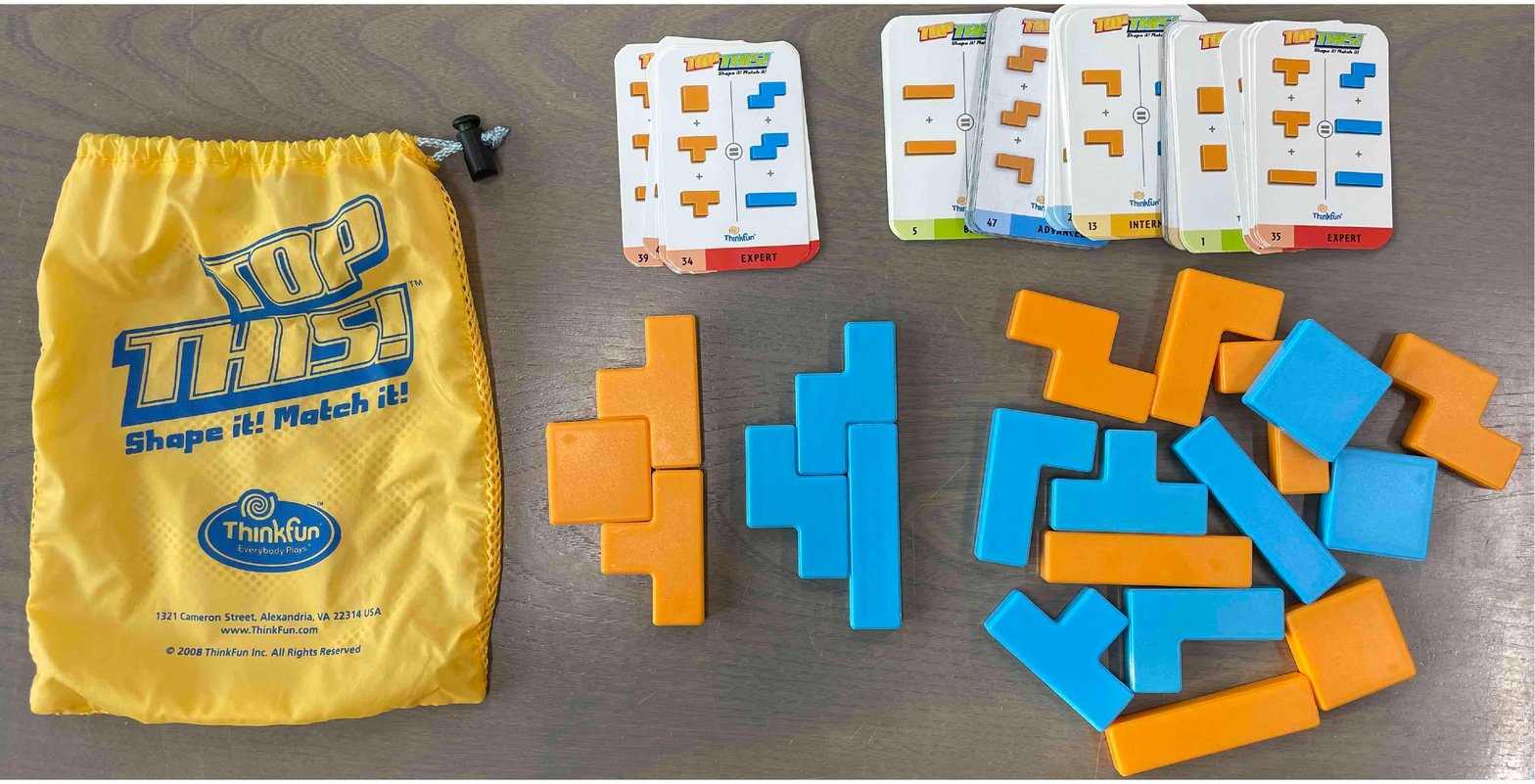}
\caption{Shape Logic (commercial product by ThinkFun)}
\label{fig:shapelogic}
\end{figure} 

In this study, we focus on such a puzzle that is known as the \emph{common shape puzzle}.
Many instances of this puzzle are available in puzzle society, and 
a commercial product named ``Shape Logic'' exists.
(The authors confirmed that this puzzle was named ``Top This!'' in 2008 (\figurename~\ref{fig:shapelogic}),
``ShapeOmetry'' in 2012, and ``Shape Logic'' more recently by the same puzzle maker.
The puzzle ``Top This!'' won three awards in 2008.\footnote{\url{https://www.thinkfun.com/about-us/awards/}}
However, in this paper, we use the most recent name.)
In the shape logic puzzle, we are provided with two sets $\calS_1$ and $\calS_2$ of polygons.
We must find a polygon $X$ that can be formed by the pieces in $\calS_1$ and $\calS_2$, respectively, 
as in the classic silhouette puzzle Tangram.
The main difference between the Tangram and the shape logic puzzle is that 
the target shape $X$ is not provided, which drastically increases the difficulty of the puzzles.

Hereafter, we suppose that $\max\{\msize{\calS_1}/\msize{\calS_2},\msize{\calS_2}/\msize{\calS_1}\}$ is bounded above by a constant.
(We note that if $\calS_1$ contains only one piece, the target shape $X$ is fixed to it.
Therefore, it is equivalent to the classic puzzle Tangram for $\calS_2$, and 
it is NP-complete even if all pieces in $\calS_1$ and $\calS_2$ are rectangles \cite{DemaineDemaine2007}.)

We first demonstrate that it is polynomial-time solvable if $\msize{\calS_1}+\msize{\calS_2}$ is a constant.
Subsequently, we show that the shape logic puzzle is strongly NP-complete even if all the pieces in $\calS_1$ and $\calS_2$ are small rectangles.
We state that a rectangle in $\calS_1 \cup \calS_2$ is small if it has a polynomial size of $\msize{\calS_1}+\msize{\calS_2}$.

\begin{figure}[thb]\centering
\includegraphics[bb=0 0 1250 1250,width=0.4\textwidth]{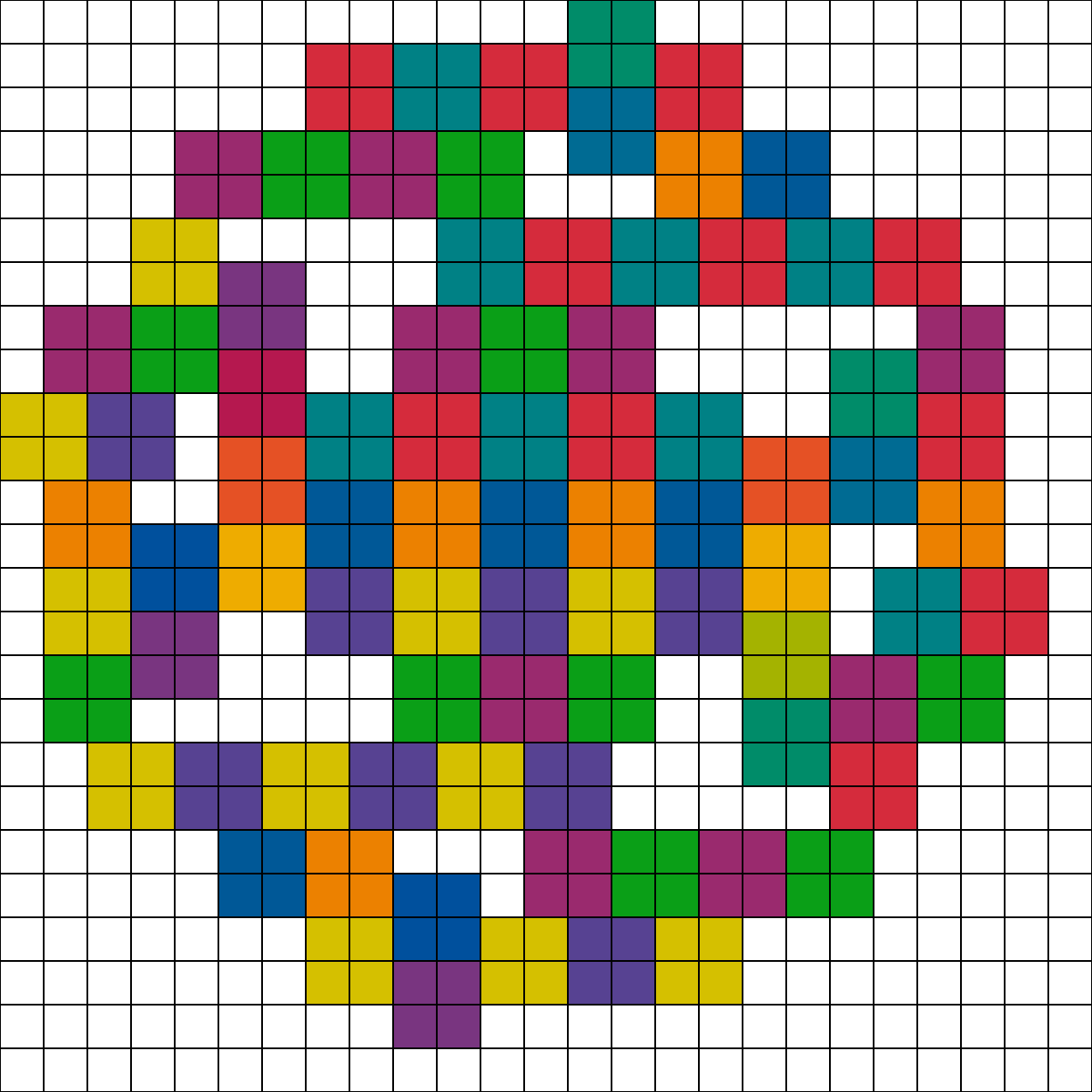}
\includegraphics[bb=0 0 1250 1250,width=0.4\textwidth]{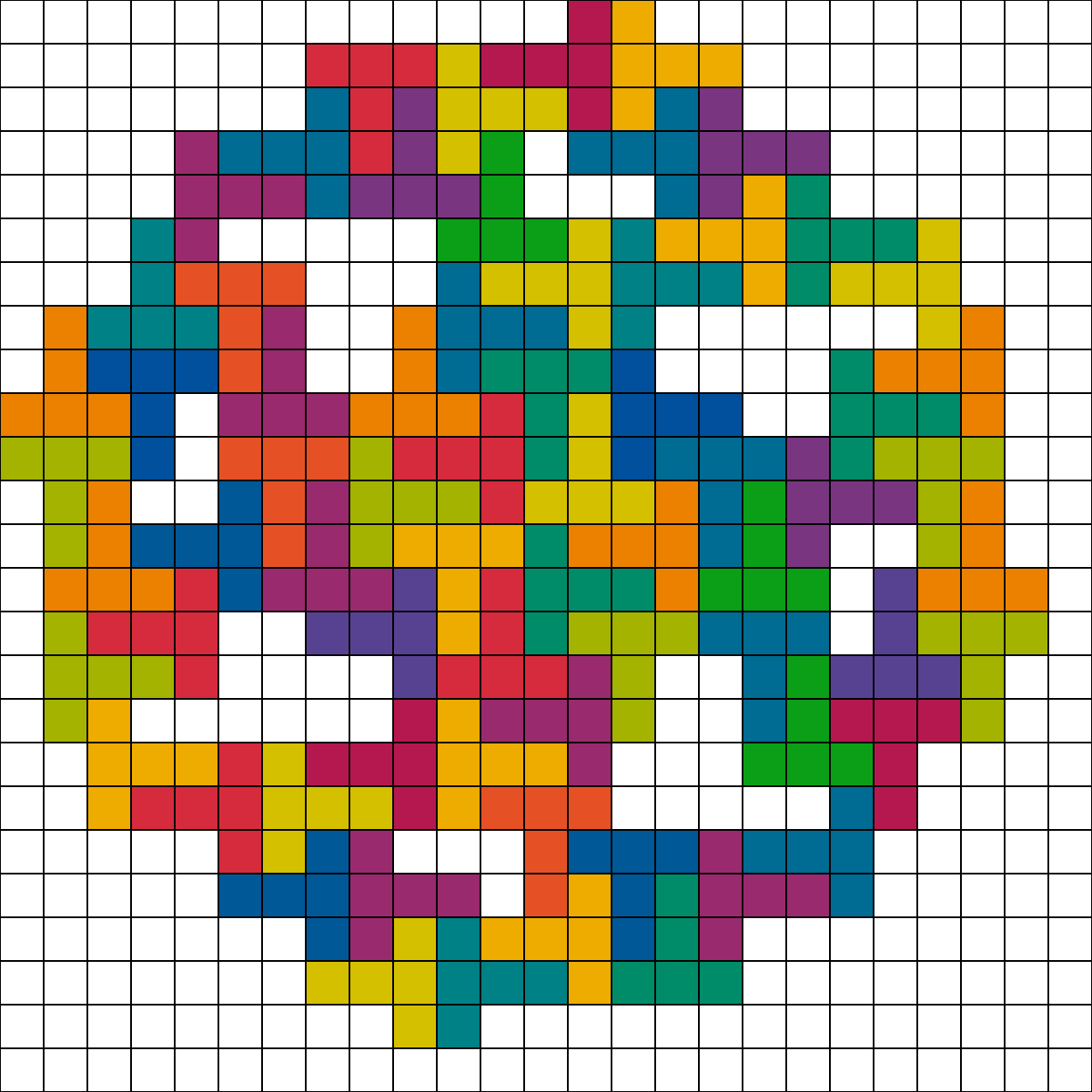}
\caption{Copies of a pentomino and copies of a tetromino share a large common shape}
\label{fig:5T=4Q}
\end{figure} 

Next, we focus on a similar puzzle named the \emph{common multiple shape puzzle}, 
which has been investigated in puzzle society for a long time under a few different names 
such as ``polypolyomino''\footnote{\url{https://www.iread.it/Poly/}} and  
``polyform compatibility''\footnote{\url{https://sicherman.net/polycur.html}}.
We propose the term ``the (least) common multiple shape'' based on the term
``the least common multiple,'' as the corresponding Japanese name is used in Japanese puzzle society.\footnote{In Japan,
we use \raisebox{-0.4em}{\includegraphics[width=48pt]{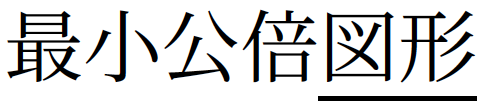}} %最小公倍\underline{図形}
(least common multiple \underline{shape}) following 
\raisebox{-0.4em}{\includegraphics[width=40pt]{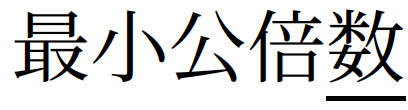}} % 最小公倍\underline{数}
(least common multiple \underline{number}).}
The representative instance of this puzzle is as follows: We are provided with two polygons $P$ and $Q$.
The puzzle asks us to find the (smallest) shape $X$ that can be tiled by $P$, and also tiled by $Q$.
That is, we can use any number of copies of $P$ and $Q$, and find the common shape $X$ that can be filled by only copies of $P$,
as well as only copies of $Q$. 
The goal of this puzzle is to determine the minimum shape, however, 
it is known that some pairs result in a huge solution (e.g., \figurename~\ref{fig:5T=4Q}), 
and at times, it cannot be guaranteed that this is the minimum shape for such a huge solution.
(The problem in \figurename~\ref{fig:5T=4Q} of finding a small common multiple shape of the T-pentomino and O-tetromino 
was first proposed by Robert Wainright as a problem at the conference of games and puzzles 
competitions on computers\footnote{\url{http://hp.vector.co.jp/authors/VA003988/gpcc/gpcc.htm}} in 2005 and 2011. 
A solution with an area of 600 was found in 2011, and 
it was improved to 340 in 2011.\footnote{\url{http://deepgreen.game.coocan.jp/MCFG/MCFG\_{}index.htm}}
In fact, it remains open whether or not this shape with an area of 340 in \figurename~\ref{fig:5T=4Q} is the smallest.)

We naturally consider the (least) common multiple shape variant of the shape logic puzzle.
That is, for given sets $\calS_1$ and $\calS_2$ of polygons, the puzzle asks us to find a small common shape $X$
that can be filled by copies of pieces in $\calS_1$ (and $\calS_2$, respectively).
We show that this puzzle is undecidable even if each set of $\calS_1$ and $\calS_2$ contains small polyominoes.
As a corollary, we also demonstrate that the following problem is undecidable:
For a given set $\calS$ of small polyominoes, determine whether a rectangle can be formed using copies of the pieces in $\calS$.

In this study, we also present a formulation of these puzzles and verify the feasibility with a computer.
We recently discovered that such puzzles can be solved by SAT-based solvers with sophisticated modeling 
far more efficiently than when using other methods \cite{BHHM2021}.
By determining an efficient formulation of this puzzle and using a SAT-based solver, 
we also improve several known instances of the common multiple shapes that have been investigated in puzzle society.

\section{Preliminaries}
\label{sec:pre}
A \emph{polyomino} is a polygon that can be obtained by joining one or more unit squares edge to edge \cite{Golomb1994}.
In this study, we only consider simple polyominoes (without holes) as polygons.
(We note that even if all pieces are simple, the solution may have holes, as indicated in \figurename~\ref{fig:5T=4Q}.)
If a polyomino $P$ is formed by $k$ unit squares, we refer to $P$ as a \emph{$k$-omino}.
For a specific $k$, we also refer to it as a \emph{monomino}, \emph{domino}, \emph{tromino}, \emph{tetromino}, \emph{pentomino}, 
and \emph{hexomino} for $k=1,2,3,4,5,6$, respectively.

A set $\calS_1$ of polyominoes is said to be a set of \emph{small} polyominoes 
when the maximum polyomino in $\calS_1$ is a $k$-omino for $k=O(\msize{\calS_1}^c)$ for a positive constant $c$.
In this case, we assume that the input size of the problem is bounded above by $O(p(n))$ for a polynomial function $p$, 
where $n=\msize{\calS_1}+\msize{\calS_2}$.

In this study, we consider two problems on polyominoes.
The first problem is the \emph{shape logic puzzle}.
Given two sets of $\calS_1$ and $\calS_2$ of small polyominoes, 
the puzzle asks us to form a common polyomino $X$ using all pieces in $\calS_1$ and all pieces in $\calS_2$, respectively.
The goal shape $X$ is not provided.
Clearly, the shape logic puzzle is in NP when all pieces are small 
as we can guess $X$ and verify the feasibility of the given packing of $\calS_1$ and $\calS_2$ on $X$ in polynomial time.

The second problem is the \emph{common multiple shape puzzle}.
Given two finite sets $\calS_1$ and $\calS_2$ of polyominoes, 
the puzzle asks us to form a common polyomino $X$ with a positive area 
using copies of the pieces in $\calS_1$ and copies of the pieces in $\calS_2$, respectively.
This puzzle generalizes both of the shape logic puzzle and the puzzle known as 
the \emph{polypolyomino} (also referred to as \emph{polyform compatibility}).
The latter puzzle is the case in which $\msize{\calS_1}=\msize{\calS_2}=1$.
It can be extended from two sets to three or more sets naturally. (See Section \ref{sec:implementation} in this case.)

For a finite set $\calS$ of polyominoes, we define a set $\hat{\calS}$ of polyominoes $P$ 
such that $P$ can be formed by copies of the pieces in $\calS$.
Clearly, $\hat{\calS}$ is infinite and countable.
That is, the common multiple shape puzzle asks whether or not $\hat{\calS_1}\cap\hat{\calS_2}\neq\emptyset$.
When $\hat{\calS_1}\cap\hat{\calS_2}\neq\emptyset$, 
we refer to an element in $\hat{\calS_1}\cap\hat{\calS_2}$ as a \emph{common multiple shape}.
Among the common multiple shapes, the smallest one is the \emph{least common multiple shape}.

\section{Complexity of Shape Logic Puzzle}
\label{sec:shape-logic}

In this section, we focus on the generalized shape logic puzzle.
That is, given two sets $\calS_1$ and $\calS_2$ of polyominoes,
we need to decide if all pieces in $\calS_1$ (and in $\calS_2$) can form a common polyomino $X$.

\begin{obs}
\label{obs:poly}
When $\msize{\calS_1}+\msize{\calS_2}$ is a constant $k$, 
the generalized shape logic puzzle can be solved in polynomial time of $n$, 
where $n$ is the total number of vertices in $\calS_1\cup\calS_2$.
\end{obs}
\begin{proof}(Outline.)
We solve the problem by brute force using the same %proof 
technique as in \cite[Section 3]{DKKMORRUU2020}.
In \cite[Section 3]{DKKMORRUU2020}, they %the authors 
presented a method for solving the symmetric assembly puzzle,
which asks us to form a symmetric shape by using the pieces in a set of (general) simple polygons in polynomial time.

The generalized shape logic puzzle can be reduced to the symmetric shape puzzle in polynomial time as follows:
Suppose that the generalized shape logic puzzle has a solution and 
the pieces in $\calS_1$ forms a polygon $P$, which can also be formed by the pieces in $\calS_2$.
Then, without loss of generality, $P$ can be placed so that its rightmost vertex $v$ is on $\partial P$;
that is, any point in $P$ is not right of $v$.
At this point, we obtain a symmetric shape by joining $P$ and its mirror image $P^R$ at vertex $v$ with its mirror image on $\partial P^R$.
That is, when the shape logic puzzle with $\calS_1$ and $\calS_2$ has a solution,
the symmetric shape puzzle also has a solution for $\calS_1\cup\calS_2$ such that the left half of the symmetric shape consists of 
the pieces in $\calS_1$ and the right half of the symmetric shape consists of the pieces in $\calS_2$.
The proof in \cite[Section 3]{DKKMORRUU2020} is based on brute force.
Therefore, we can restrict our search to a symmetric shape that also provides a solution for the shape logic puzzle in $\calS_1\cup\calS_2$.
As the original brute force algorithm for the symmetric shape puzzle runs in a polynomial time of $n$, so does our algorithm.
\qed\end{proof}

\begin{theorem}
The shape logic puzzle is strongly NP-complete even if all pieces in $\calS_1$ and $\calS_2$ are small rectangles.
\end{theorem}

\begin{proof}
According to the definition of a small polyomino, the problem is in NP.
Thus, we demonstrate the NP-hardness using a reduction from the 3-partition problem.
In the 3-partition problem, we are provided with a multiset of 
$3m$ positive integers $A=\{a_1,a_2,\ldots,a_{3m}\}$, where the $a_i$s are bounded 
above by a polynomial of $m$. 
The goal is to partition the multiset $A$ into $m$ triples such that every triple
has the same sum $B=(\sum_{i=1}^{3m}a_i)/m$.
It is known that the 3-partition problem is strongly NP-complete even if every $a_i$ satisfies $B/4<a_i<B/2$  \cite[SP16]{GJ79}.
Without loss of generality, we assume that $a_i>3$ for each $i$ and $B=3mB'$ for a positive integer $B'$.
%and $m=2m'$ for some positive integers $B'$ and $m'$.
%
Then, the set $\calS_1$ consists of $3m$ rectangles of size $1\times (a_i+3m^2)$ for each $i=1,2,\ldots,3m$.
Furthermore, $\calS_2$ consists of $3m$ congruent rectangles of size $m\times (B/(3m)+3m)$.
The construction can be completed in a polynomial time of $m$.

Subsequently, we observe that $3m<B/(3m)+3m<B/4+3m^2<a_i+3m^2$ for each $i$,
as $3m$ pieces exist in $\calS_1$, $a_i>3$, and $a_i>B/4$.
That is, (1) $B/(3m)+3m$ is larger than $3m$, which is the number of long and slender rectangles in $\calS_1$, and
(2) each $a_i+3m^2$ cannot fit into any rectangle that is formed by the pieces in $\calS_2$ except
if a rectangle with a width of $m$ is created.
Therefore, the only means of forming the same shape using the pieces in $\calS_1$ and the pieces in $\calS_2$
is to form a rectangle with a size of $m\times(B+9m^2)$ using the pieces in $\calS_2$
and to pack long rectangles with a size of $1\times (a_i+3m^2)$ into this frame.
The arrangement of pieces in $\calS_1$ directly provides the solution to the original instance of the 3-partition problem.
\qed\end{proof}

\section{Undecidability of Common Multiple Shape Puzzle}
\label{sec:undecidable}

In this section, we demonstrate that the common multiple shape puzzle is undecidable.

\subsection{Undecidability of a generalized jigsaw puzzle}

We first consider a generalized jigsaw puzzle.
We borrow several notions from \cite{BDDHMY2017}. Each piece is a square with four edges and has its own color.
We denote the set of colors as $C=\{0,1,2,\ldots,c,\bar{1},\bar{2},\ldots,\bar{c}\}$.
In our jigsaw puzzle, we tile the pieces into a rectangular frame so that
each edge is shared by two adjacent pieces with colors $i$ and $\bar{i}$, except for the boundary of the frame.
A special color $0$ exists, which should match to the frame. 
That is, when we tile the pieces, the outer boundary has the color $0$, and no inside edge has the color $0$.

In our jigsaw puzzle, we are allowed to use copies of a piece in $\calS$ multiple times, 
which is the significant difference between our puzzle and that in \cite{BDDHMY2017}.
Therefore, for a given finite set $\calS$, we have infinitely countable means of tiling the pieces.
Subsequently, the jigsaw puzzle problem is defined as follows:

\begin{description}
\item[Input:]  A set $\calS$ of unit square pieces such that each piece has four colors in $C$ on its four edges.
\item[Output:] Decide if there is a polyomino region $R$ such that 
$R$ can be tiled by copies of pieces in $\calS$ in which each inner edge is shared by 
two adjacent pieces of colors $i$ and $\bar{i}$ (with $i>0$), and each edge on the boundary $\partial R$ has the color $0$.
\end{description}

We first present the following lemma.

\begin{lemma}
\label{lem:rect}
There exists a finite set $\calS$ of jigsaw puzzle pieces such that
the area $R$ is tiled by copies of the pieces in $\calS$ with the boundary color $0$ along $\partial R$
if and only if $R$ is a rectangle with a size of at least $3\times 3$.
\end{lemma}

\begin{figure}[thb]\centering
\includegraphics[bb=0 0 161 161,width=0.25\textwidth]{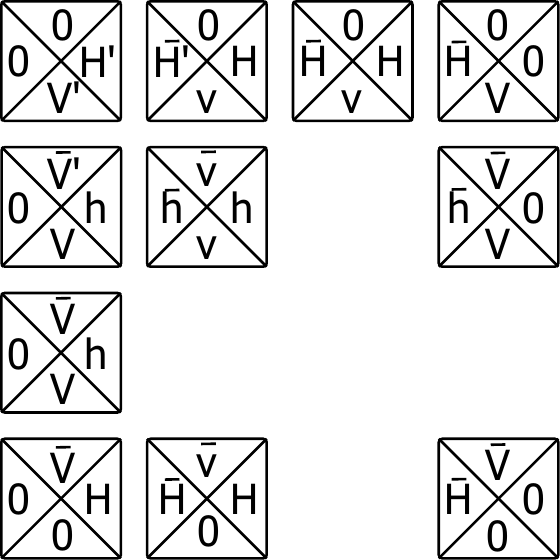}
\caption{Jigsaw puzzle in a rectangle}
\label{fig:rect}
\end{figure} 

\begin{proof}
We consider the set of 11 jigsaw puzzle pieces depicted in \figurename~\ref{fig:rect}.
For ease of reference, we use some letters such as $H$, $V$, etc.~instead of the numbers $1$, $2$, etc.~in the figure.
As every piece contains $H$, $\bar{H}$, $H'$, $\bar{H'}$, $h$, or $\bar{h}$,
without loss of generality, we can assume that one piece is placed so that its $\bar{H}$, $\bar{H'}$, or $\bar{h}$ is on its left,
and all the pieces are then aligned in the same direction, as indicated in the figure.
The boundary of the jigsaw puzzle is labeled by $0$.
We observe that we cannot form a rectangle with a size of $2\times n$ (and $n\times 2$) for any $n$ 
because $V'$ and $\bar{V}$ (and $H'$ and $\bar{H}$) do not match.
Furthermore, we observe that an edge is colored by $h$ or $v$ if and only if it is not incident to a vertex on the boundary.
Therefore, we have no other means of forming a rectangle by filling the copies of the boundary pieces with the color $0$.
In particular, we cannot create a corner with the angle $270^\circ$; to achieve this, we must place one ``corner boundary'' piece inside,
which results in the color $0$ being inside the shape, and this is not permitted.
\qed\end{proof}

We show that our jigsaw puzzle problem is undecidable.

\begin{lemma}
\label{lem:tile}
There exists a finite set $\calS$ of pieces of the jigsaw puzzle such that the jigsaw puzzle is undecidable.
\end{lemma}
\begin{proof}(Outline.)
We present a polynomial-time reduction from the following Post correspondence problem:
\begin{description}
\item[Input:] A sequence of pairs of strings $s_1=(t_1;b_1),s_2=(t_2;b_2),\ldots,s_{n}=(t_{n},b_{n})$.
\item[Output:] For a pair $s_i=(t_i;b_i)$ of strings, we define $T(s_i)=t_i$,$B(s_i)=b_i$.
	   We decide if there exists a sequence of pairs $s_{i_1},s_{i_2},s_{i_3},\ldots,s_{i_k}$ of strings
	   such that $T(s_{i_1})T(s_{i_2})T(s_{i_3})\cdots T(s_{i_k})=B(s_{i_1})B(s_{i_2})B(s_{i_3})\cdots B(s_{i_k})$.
\end{description}
Let $\Sigma$ be an alphabet, namely the set of letters that is used in the sequence.
We note that we can use each pair $s_i$ can be used any number of times.
It is well known that the Post correspondence problem is undecidable even if $\msize{\Sigma}$ is a constant \cite{Post1946}.

\begin{figure}[thb]\centering
\includegraphics[bb=0 0 540 375,width=0.9\textwidth]{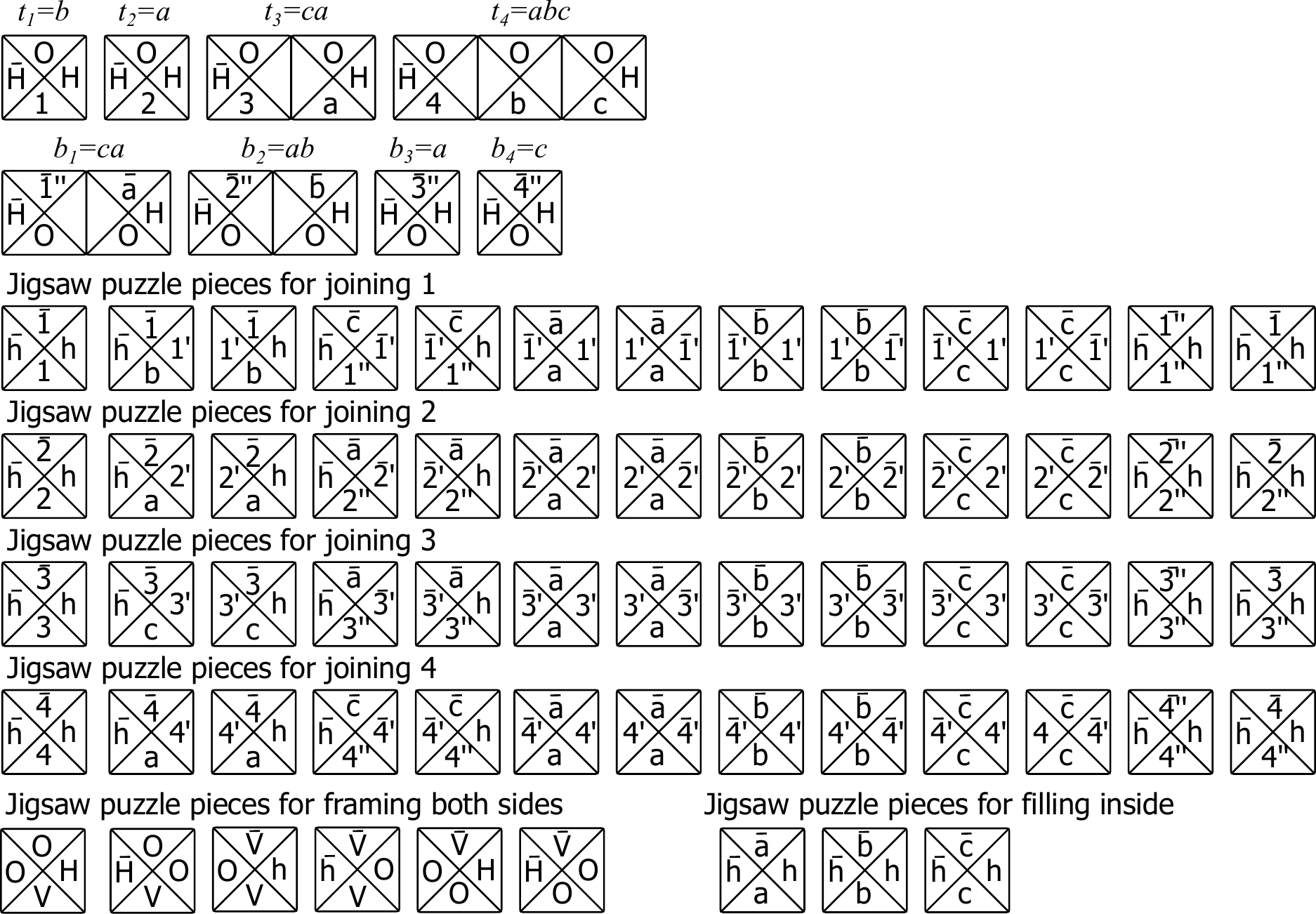}
\caption{A reduction from the Post correspondence problem to the jigsaw puzzle problem}
\label{fig:post-sample}
\end{figure} 

We demonstrate the reduction by using a concrete example $\Sigma=\{a,b,c\}$, $s_1=(b;ca),s_2=(a;ab),s_3=(ca;a),s_4=(abc;c)$ (\figurename~\ref{fig:post-sample}).
We prepare one piece, one piece, two pieces, and three pieces of jigsaw puzzle for each string 
$t_1=b$, $t_2=a$, $t_3=ca$, and $t_4=abc$, respectively.
We set each string to be uniquely constructible: two pieces for $t_3=ca$ have their own color (distinct from any other color) between them,
and three pieces for $t_4=abc$ have their own two colors between $a$ and $b$, and $b$ and $c$ (these are blank in \figurename~\ref{fig:post-sample}).
The top color is $0$, the left color is $\bar{H}$, and the right color is $H$ for each piece.
(As in the proof of Lemma \ref{lem:rect}, we regard certain letters as numbers greater than $n$.)
Hereafter, we consider these strings as rectangular pieces that are represented by sizes of $1\times 1$, $1\times 1$, $2\times 1$, and $3\times 1$, respectively.
The leftmost bottom color of the rectangular piece $t_i$ is the color $i$, which is referred to as the \emph{ID} of this rectangular piece.
The color of the other edge corresponds to the letter in the string.
That is, the second and the third pieces of the rectangle representing $t_4=abc$ have the colors $b$ and the color $c$, respectively.
(We regard these letters as unique numbers in the color set $C$.
As the size of the alphabet $\Sigma$ is a constant, regarding these letters as numbers has no influence on our arguments.)

Subsequently, we prepare two, two, one, and one pieces for the strings $b_1=ca$, $b_2=ab$, $b_3=a$, and $b_4=c$, respectively.
As with the strings $t_i$, $b_1,b_2,b_3,b_4$ each corresponds to a rectangular piece with a size of $2\times 1$, 
$2\times 1$, $1\times 1$, and $1\times 1$, respectively.
The bottom color is $0$, the left color is $\bar{H}$, and the right color is $H$ for these rectangular pieces.
The top colors of the rectangular piece are represented by the letters, except for the leftmost edge,
which has the color ID $\bar{i''}$ for the string $b_i$.

Next, we prepare to join two pieces with the IDs $i$ and $i''$.
Hereafter, we use $(c_u,c_b,c_l,c_r)$ to denote the top color $c_u$, bottom color $c_b$, left color $c_l$, and right color $c_r$ of a piece. 
Furthermore, we assume that the top letter of $t_i$ is $x_i$ and the top letter of $b_i$ is $y_i$.
We first prepare a piece with colors $(\bar{i},i,\bar{h},h)$, which is a wire of the ID in the vertical direction.
We also prepare a piece with colors $(\bar{i},x_i,\bar{h},i')$, which turns the ID to the right,
and a piece with colors $(\bar{i},x_i,\bar{i'},h)$, which turns the ID to the left.
The ID is turned to the right or left using one of these pieces and runs horizontally.
The prime symbol $'$ means that the ID turns once.
Thereafter, we prepare two pieces with the colors $(\bar{y_i},i'',\bar{h},i')$ and $(\bar{y_i},i'',\bar{i'},h)$ to turn the ID downwards.
In this case, the symbol $''$ means that the ID turns twice.
Furthermore, we prepare a piece with the color $(\bar{j},j,\bar{i'},i')$ for each letter $j\in\Sigma$ to propagate 
the ID in the horizontal direction.
We also add a piece with the color $(\bar{i''},i'',\bar{h},h)$ to pass the ID downwards after turning twice.
In a special case, an ID can directly move from top to bottom without turning.
We prepare a piece with the color $(\bar{i},i'',\bar{h},h)$ to deal with this case.
Thus, we prepare a total of eight pieces for the IDs $i$ and $i''$.
% 8n pieces 

Subsequently, we prepare pieces to form the left and right sides of the rectangular frame.
We prepare six pieces with the colors 
$(0,V,0,H)$, $(0,V,\bar{H},0)$, $(\bar{V},V,0,h)$, $(\bar{V},V,\bar{h},0)$, $(\bar{V},0,0,H)$, and $(\bar{V},0,\bar{H},0)$.
% 6 pieces.
Finally, we prepare pieces $(\bar{j},j,\bar{h},h)$ for each $j\in \Sigma$ to fill the holes in the frame.

We prepare $\sum_{i=1}^n (\msize{t_i}+\msize{b_i})+8n+6+\msize{\Sigma}$ pieces in total.
Therefore, the jigsaw puzzle can be constructed in polynomial time for 
the size of a given instance of the Post correspondence problem.

We demonstrate that the instance $s_1=(t_1;b_1),s_2=(t_2;b_2),\ldots,s_{n}=(t_{n},b_{n})$ of the Post correspondence problem 
has a solution if and only if the jigsaw puzzle has a solution such that a rectangular area $R$ is filled with copies of 
the pieces with the color $0$ only on $\partial R$.

\begin{figure}[t]\centering
\includegraphics[bb=0 0 345 619,width=0.6\textwidth]{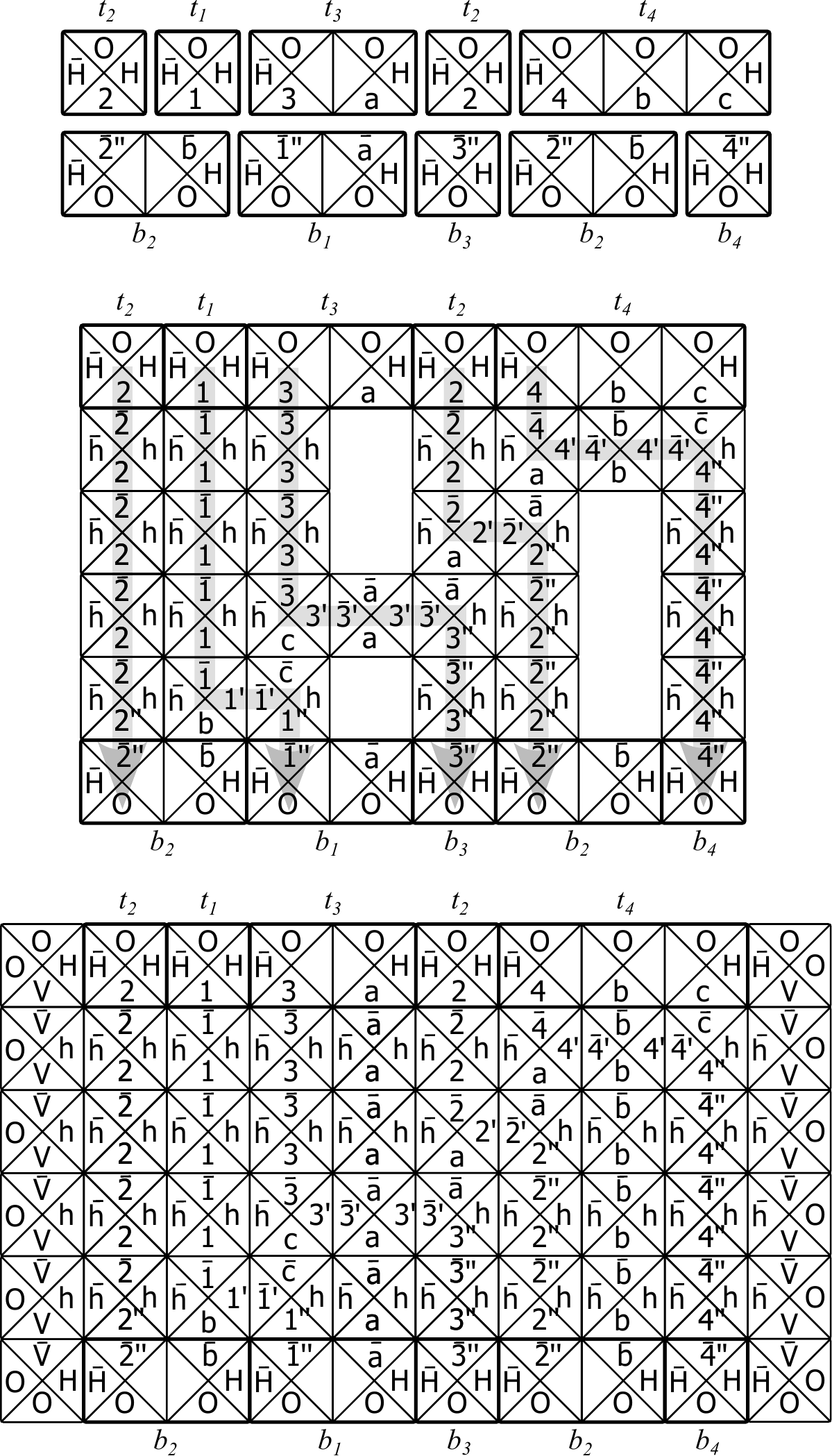}
\caption{Construction of a solution of the jigsaw puzzle from a solution of the instance of the Post corresponding problem}
\label{fig:post-sample-c}
\end{figure} 

We first assume that the sequence $s_{i_1},\ldots,s_{i_k}$ is a solution, from which
we construct a rectangular shape using the set of pieces of the jigsaw puzzle.
We use $s_1=(b;ca),s_2=(a;ab),s_3=(ca;a),s_4=(abc;c)$ as an example (\figurename~\ref{fig:post-sample-c}).
Intuitively, we verify that two corresponding IDs are joined by a zig-zag path with two (or zero) turns,
these zig-zag paths do not cross one another, and the corresponding letters are joined by vertical matching pieces.

We first align the rectangular pieces $t_{i_1}, t_{i_2}, \ldots, t_{i_k}$ on the top row and 
$b_{i_1}, b_{i_2}, \ldots, b_{i_k}$ on the bottom row 
following the solution $s_{i_1},\ldots,s_{i_k}$ of the Post correspondence problem.
(As a reminder, each string $t_i$ (and $s_i$) produces unique rectangular pieces.)
Thus, we obtain $0$s on the top and bottom boundaries, and 
we can join all rectangular pieces by matching $h$ and $\bar{h}$. 
Subsequently, we join all corresponding pairs of IDs using the prepared pieces.
When the gap between the top and bottom rows is sufficiently large,
each joining path for each ID $i$ can be created in one of the following manners:
\begin{enumerate}
 \item The pair of the corresponding IDs $i$ and $i''$ in the same column is directly joined vertically, 
 \item When the ID $i''$ of $b_i$ is left of the ID $i$ of $t_i$, 
       the ID $i$ first moves down vertically, turns left once, moves horizontally, turns right once, and moves downwards to the ID $i''$, and 
 \item When the ID $i''$ of $b_i$ is right of the ID $i$ of $t_i$, 
       the ID $i$ first moves down vertically, turns right once, moves horizontally, turns left once, and moves downwards to the ID $i''$.
\end{enumerate}
Any of these procedures can be performed using the pieces prepared as above.
Note that in the case (2), the first letter $x_i$ of the string $t_i$ 
appears at the corner when we use the piece $(\bar{i},x_i,\bar{i'},h)$ is used to turn left,
and the first letter $y_i$ of the string $b_i$ appears at the corner when $(\bar{y_i},i'',\bar{h},i')$ is used.
In the case (3), the pieces $(\bar{i},x_i,\bar{h},i')$ and $(\bar{y_i},i'',\bar{i'},h)$ are used for this purpose.

\clearpage

Following all the above steps, we can observe that each corresponding pair of IDs is joined by either
a straight vertical path in case (1) or a zig-zag path with two turns in cases (2) or (3).
Moreover, the $i$th letter in the common string that is produced by the sequence $s_{i_1},\ldots,s_{i_k}$ appears 
on all the horizontal edges of the $i$th piece (from left), except its boundary and pieces on the vertical line that join two corresponding IDs.
At this holds even for the holes, we can fill all of the holes using the pieces that have been prepared for filling.
Finally, we can complete the frame by arranging the pieces that have been prepared for the frame with the color $0$ on the boundary.

We assume that the jigsaw puzzle has a solution.
The pieces that correspond to $t_i$ and $b_i$ form the respective rectangles as they have their unique colors.
As all pieces have the color $h$ or $\bar{h}$, therefore, every piece is arranged so that 
$\bar{h}$ appears on the left side and $h$ appears on the right side.
Because the color $0$ matches no other colors, the rectangles for $t_i$ and the corner pieces of color $0$ on the upper edges
are arranged on the top row, as are the rectangles for $b_i$ and the corner pieces of color $0$ on the lower edges.
We need to form a rectangle using the pieces of the color $0$ on the left or right side.
(Although it may appear that we can form any polyomino other than rectangles,
we cannot create any concave corner of $270^\circ$ because an edge with the color 0 cannot be placed inside the polyomino.)

According to the color properties,
the ID color of each rectangle corresponding to $t_i$ should be connected to 
the ID color of each rectangle for $b_i$, and these $k$ paths cannot cross.
If a path has no turn, it is necessary to use some copies of the piece with the color $(\bar{i},i,\bar{h},h)$, 
one copy of the piece $(\bar{i},i'',\bar{h},h)$, and some copies of the piece $(\bar{i''},i'',\bar{h},h)$.
If a path has turns, the only possible solution is that 
the color $i$ of $t_i$ starts vertically, is changed to $i'$ after one $90^\circ$ turn,
moves horizontally, is changed to $i''$ after one $90^\circ$ turn, 
and moves down to the piece in the rectangle corresponding to $b_i$.
The colors of $t_i$ and $b_i$ appear at each turn on a horizontal edge.
Thus, the remaining holes should be packed using the pieces with the color $(\bar{j},j,\bar{h},h)$ for the matching color $j$,
and each pair of IDs of $t_i$ and $b_i$ should match.

Therefore, when the jigsaw puzzle has a solution, the pieces form a rectangle,
the pieces for $t_i$ are arranged on the top, 
the pieces for $b_i$ are arranged on the bottom, 
the corresponding pairs of IDs of $t_i$ and $b_i$ match, and 
a consistent letter is obtained along each vertical line of the pieces.
Thus, $s_i$ can be arranged following the sequence, and 
the same sequence of letters that is produced by the sequences of $t_i$ and $b_i$ can be obtained, 
which provides a solution to the Post correspondence problem, thereby completing the proof.
\qed\end{proof}

Lemmas \ref{lem:rect} and \ref{lem:tile} imply the following Theorem.
\begin{theorem}
For two finite sets $\calS_1$ and $\calS_2$ of small polyominoes,
the common multiple shape puzzle for $\calS_1$ and $\calS_2$ is undecidable.
\end{theorem}

\begin{figure}[thb]\centering
\includegraphics[bb=0 0 290 68,width=0.8\textwidth]{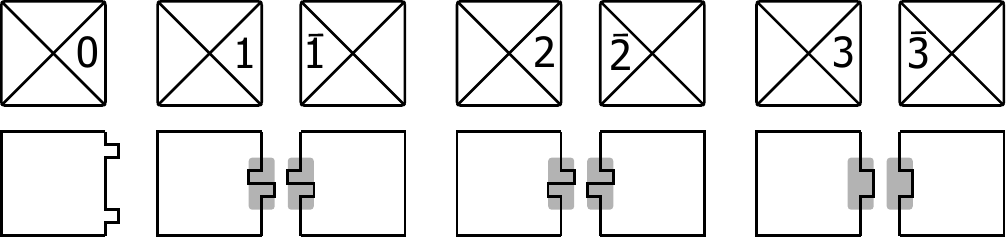}
\caption{Colored jigsaw piece for a polyomino. Each color $i$ corresponds to a zigzag pattern
that represents the integer $i$ in the binary system. The color $\bar{i}$ is its negative.}
\label{fig:jigsaw}
\end{figure}

\begin{proof}(Outline.)
We first demonstrate how to represent each piece of the jigsaw puzzle in Lemmas \ref{lem:rect} and \ref{lem:tile} using a small polyomino.
The basic concept is explained in \cite[Fig.~7]{DemaineDemaine2007}.
Each color is represented by its original zig-zag pattern.
See \figurename~\ref{fig:jigsaw} for an example of the representation.
Using the binary system, the size of the polyomino is $O((\log \msize{C})^2)$, where $C$ is the set of colors.

We consider the set $\calS_1$ of jigsaw pieces in Lemma \ref{lem:rect}
and the set $\calS_2$ of jigsaw pieces in Lemma \ref{lem:tile}.
Different colors can be used for each set, except the common color $0$.
Subsequently, according to Lemma \ref{lem:rect}, a solution to the common multiple shape puzzle is a shape that corresponds to a rectangle.
Moreover, according to Lemma \ref{lem:tile}, whether it can be constructed using the pieces in $\calS_2$ is undecidable.
The number of colors used in $\calS_1\cup\calS_2$ is linear with the size $n$ of the input.
Thus, each polyomino has an area of $O((\log n)^2)$, which means that it is small. This completes the proof.
\qed\end{proof}

\section{Improved Solutions for Common Multiple Shapes}
\label{sec:implementation}

In this section, we provide a brief formulation of generalized common shape puzzles.
The rep-tile problem, which is a type of packing puzzle on polyominoes, 
has been formulated and examined using several different computer methods \cite{BHHM2021} recently.
In \cite{BHHM2021}, the authors demonstrated that the rep-tile problem can be formulated in a natural form
that can be handled using various methods.
They compared a well-known puzzle solver, a few algorithms based on dancing links, 
an MIP solver, and a SAT-based solver with respect to for solving the packing puzzles.
In \cite{BHHM2021}, the authors concluded that the SAT-based solver is significantly faster than the other methods.
Therefore, we examined several instances 
that are available online,\footnote{\url{https://www.iread.it/Poly/} and  \url{https://sicherman.net/polycur.html}.}
 and improved some of the known results. 
The notation for small polyominoes follows ones given in these wab pages.
(They are rather trivial as you can see in the patterns.)

\begin{figure}[thb]\centering
\includegraphics[bb=0 0 950 800,width=0.3\textwidth]{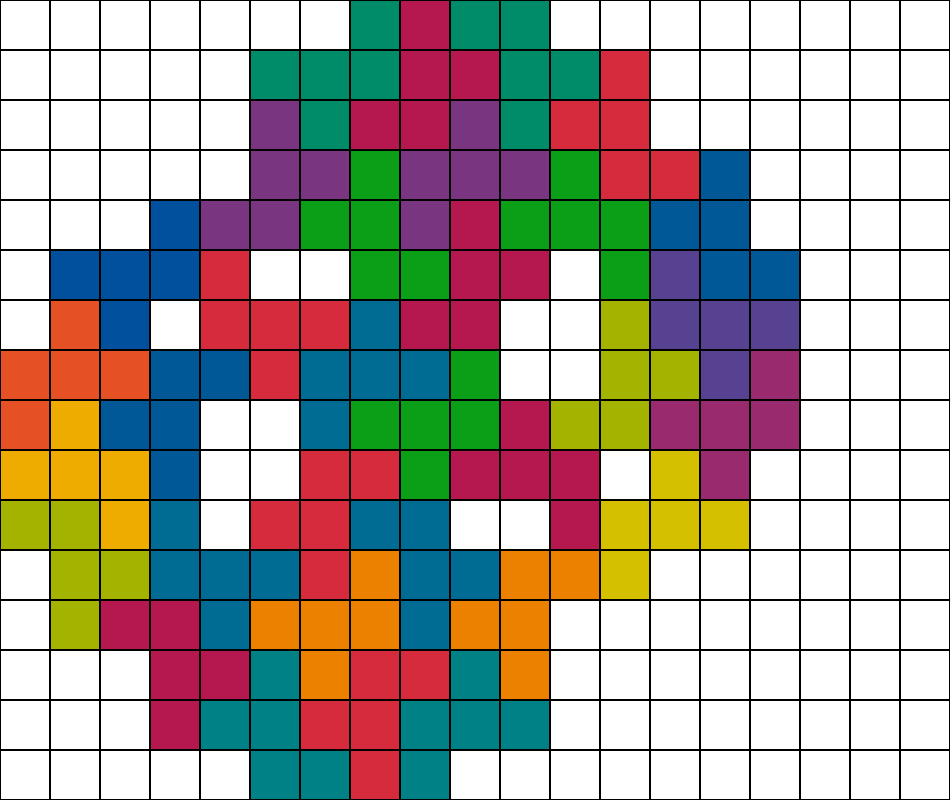}
\includegraphics[bb=0 0 950 800,width=0.3\textwidth]{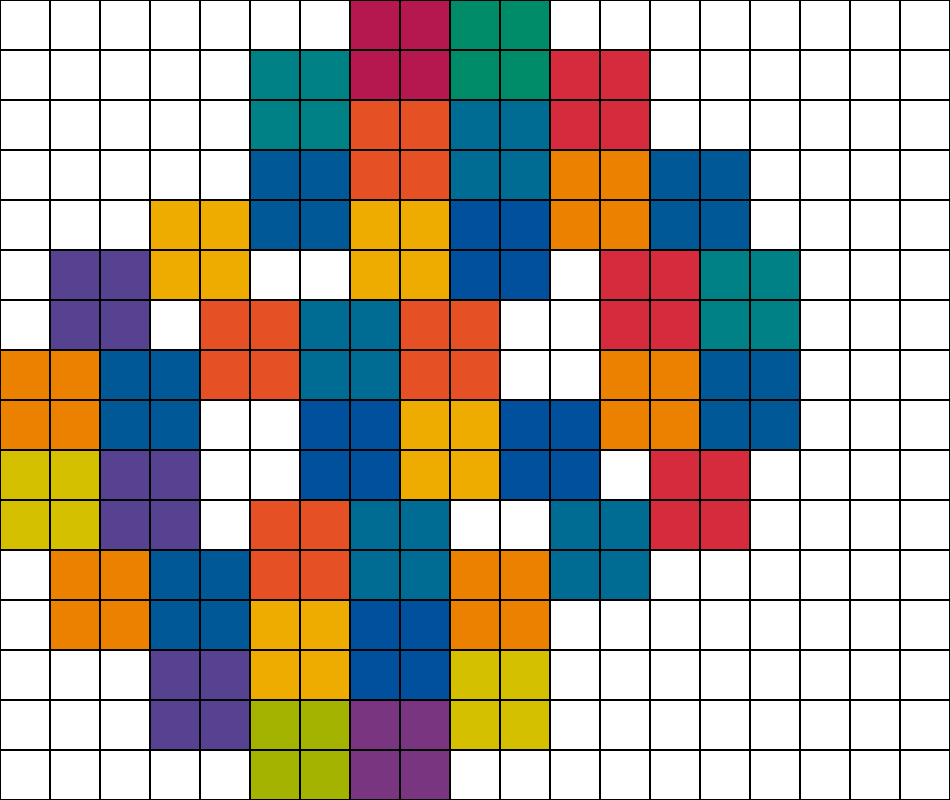}
\includegraphics[bb=0 0 950 800,width=0.3\textwidth]{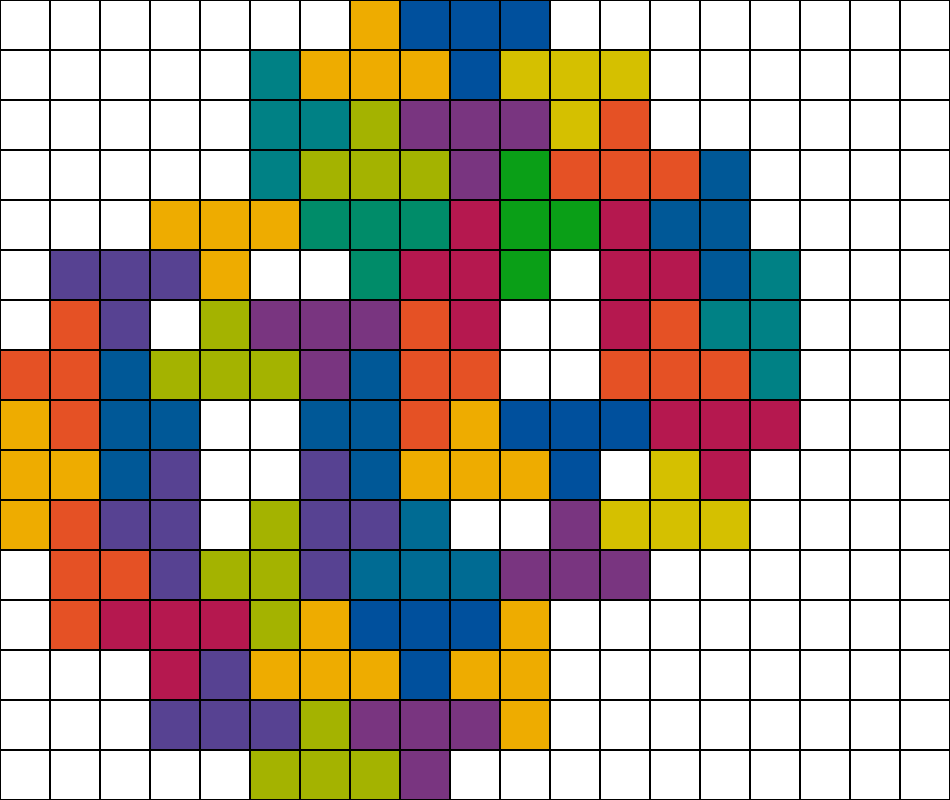}
\caption{Tiling patterns for F5Q4T4 improved from 760-omino to 160-omino}
\label{fig:F5-O4-T4}
\end{figure} 

\begin{figure}[thb]\centering
\includegraphics[width=0.3\textwidth]{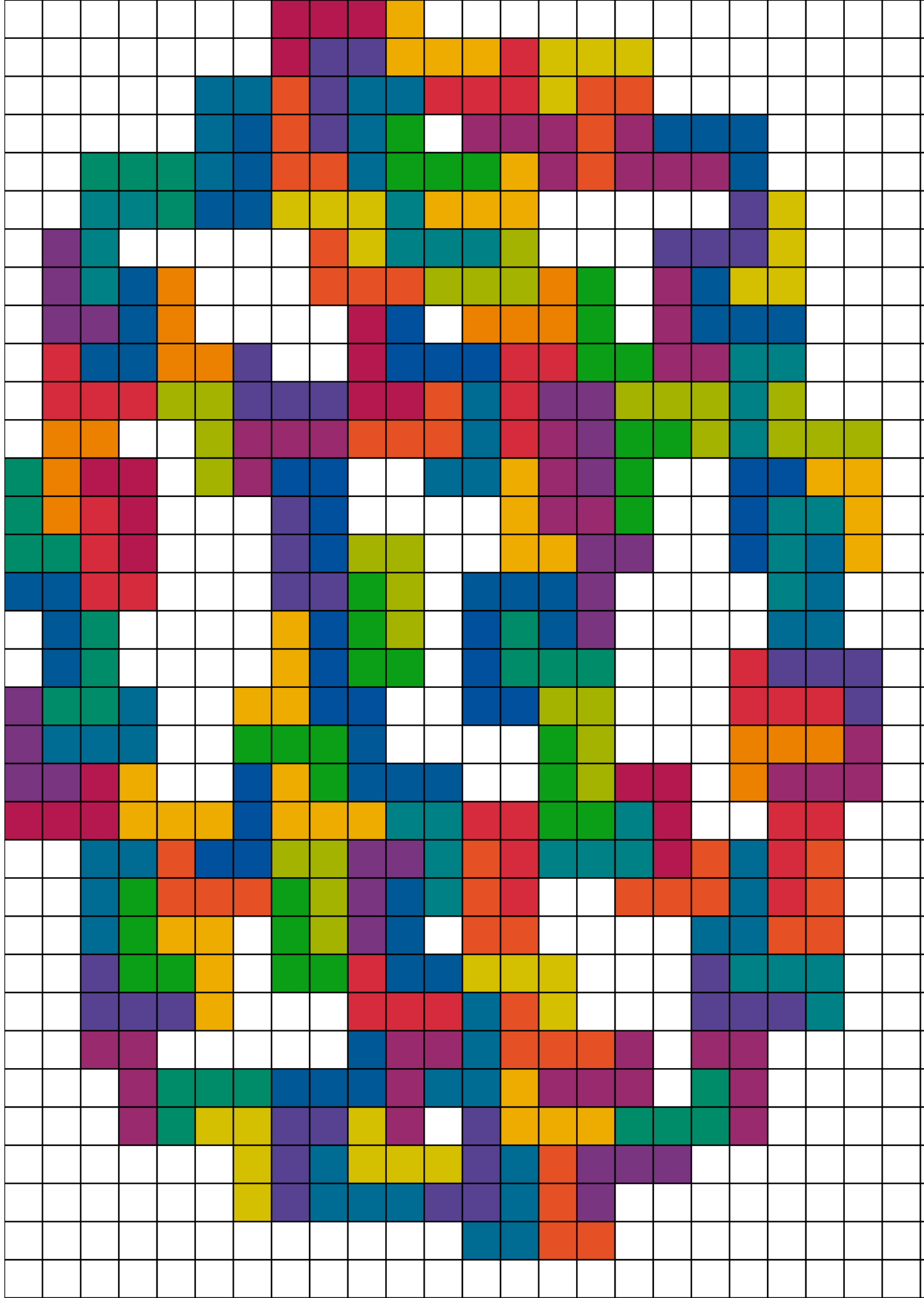}
\includegraphics[width=0.3\textwidth]{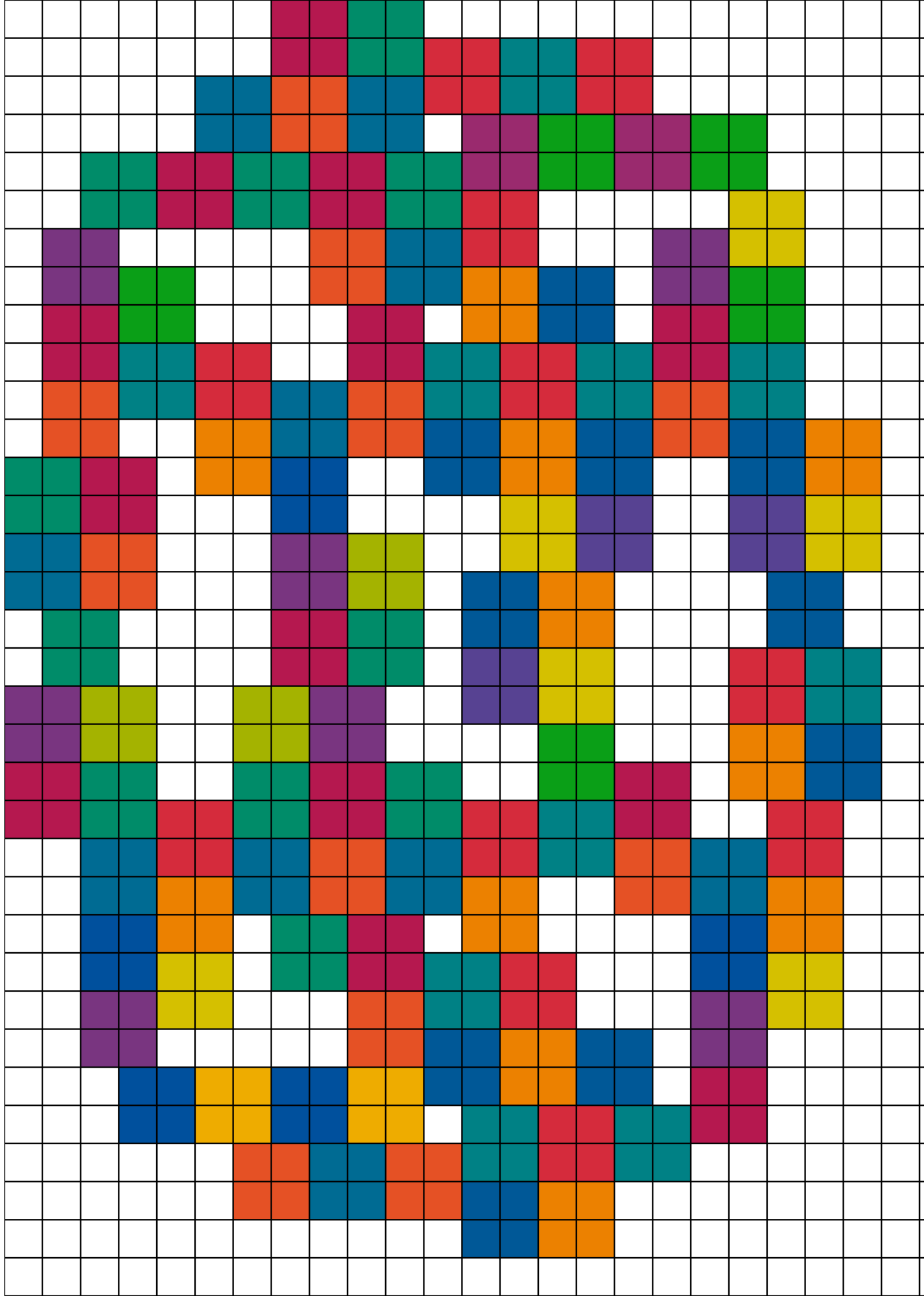}
\includegraphics[width=0.3\textwidth]{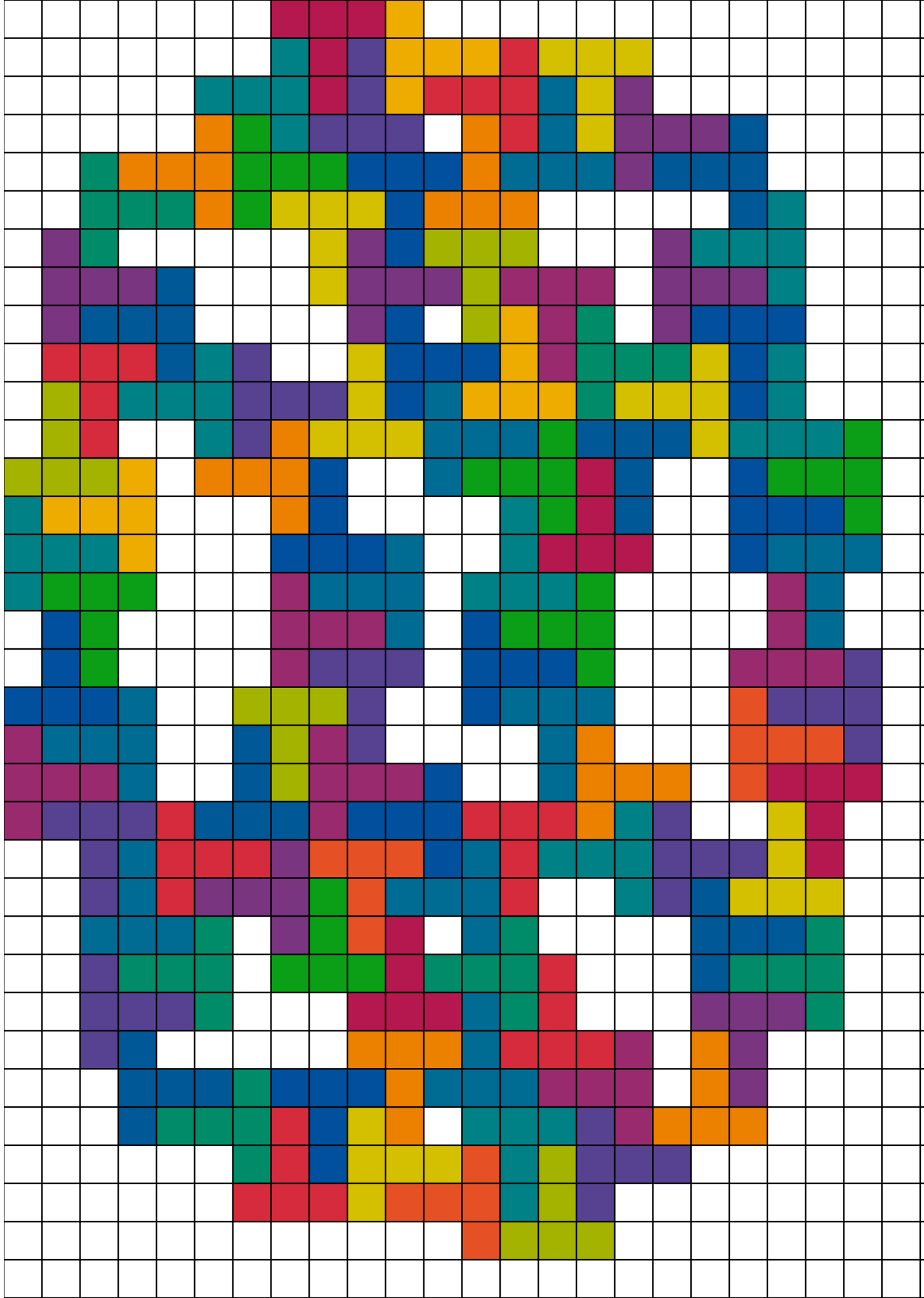}
\caption{Tiling patterns for T5L4Q4 improved from 560-omino to 480-omino}
\label{fig:T5-L4-O4}
\end{figure} 

For example, the previous best known shape for F5Q4T4 on \url{https://www.iread.it/Poly/} was a 760-omino,
and our new shape is only 160-omino (\figurename~\ref{fig:F5-O4-T4}).
The previous best known shape for T5L4Q4 on \url{https://sicherman.net/n445com/n445com.html}, 
which was a 560-omino, is improved to 480-omino (\figurename~\ref{fig:T5-L4-O4}).
%The other improved patterns are described in Appendix~\ref{sec:patterns}.

The previous best known shapes for I5P5T5, I5P5Z5, L5P5X5, and P5U5V5 on 
\url{https://sicherman.net/rosp/triplep.html} were 120-omino, 200-omino, 400-omino, and 160-omino, respectively.
We obtain new better shapes of 
110-omino for I5P5T5 (\figurename~\ref{fig:I5-P5-T5}), 
150-omino for I5P5Z5 (\figurename~\ref{fig:I5-P5-Z5}), 
360-omino for L5P5X5 (\figurename~\ref{fig:L5-P5-X5}), and
120-omino for P5U5V5 (\figurename~\ref{fig:P5-U5-V5}), respectively.

\begin{figure}[thb]\centering
\includegraphics[bb=0 0 750 750,width=0.3\textwidth]{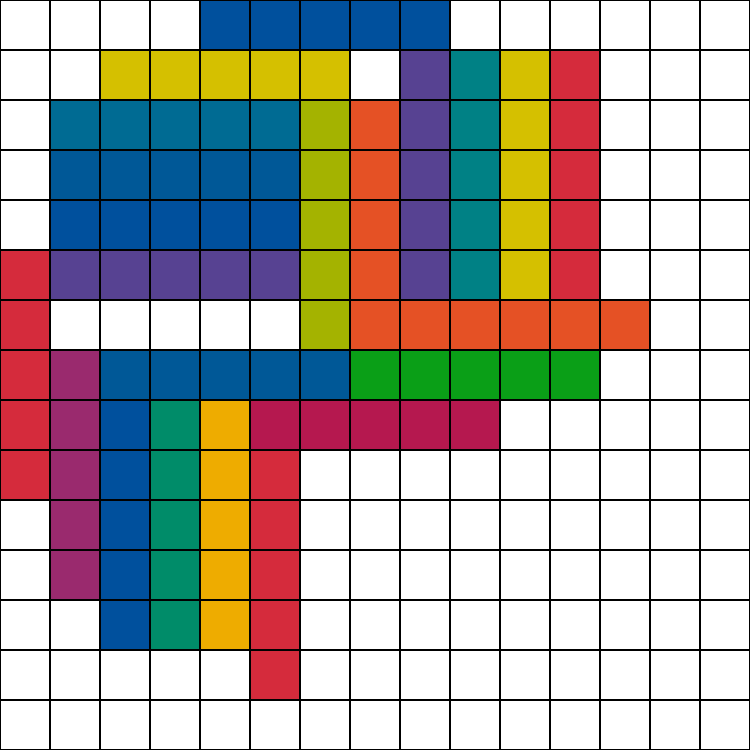}
\includegraphics[bb=0 0 750 750,width=0.3\textwidth]{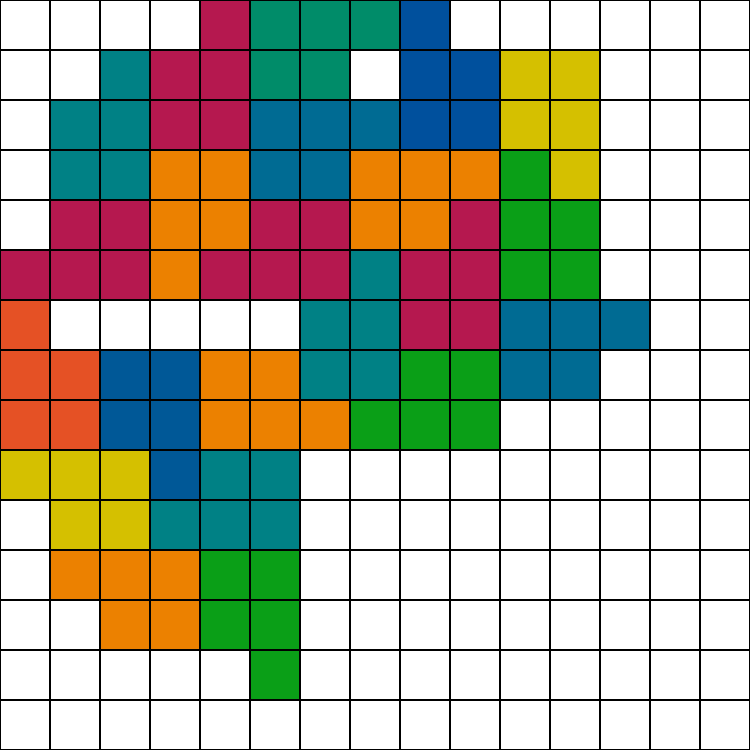}
\includegraphics[bb=0 0 750 750,width=0.3\textwidth]{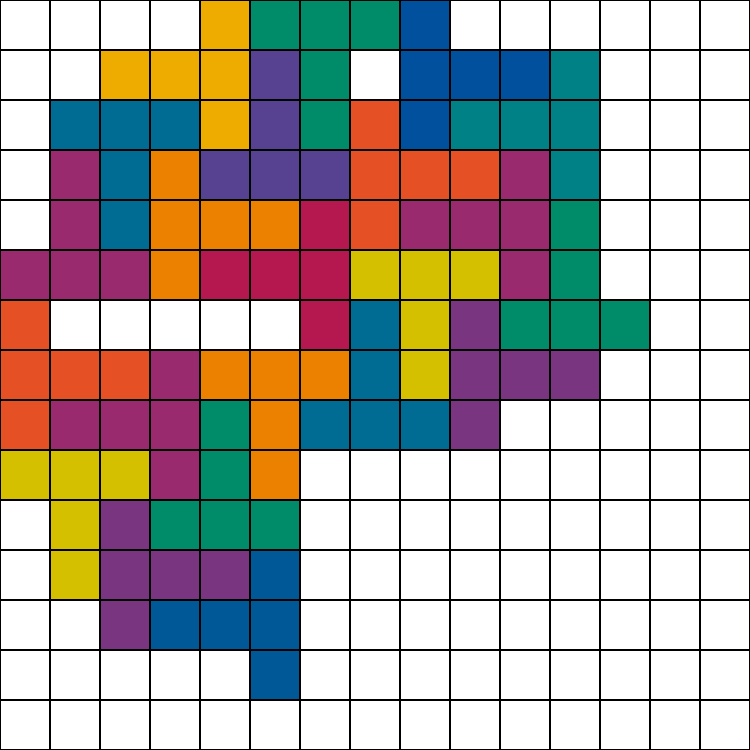}
\caption{Tiling patterns for I5P5T5 improved from 120-omino to 110-omino}
\label{fig:I5-P5-T5}
\end{figure} 

\begin{figure}[thb]\centering
\includegraphics[width=0.3\textwidth]{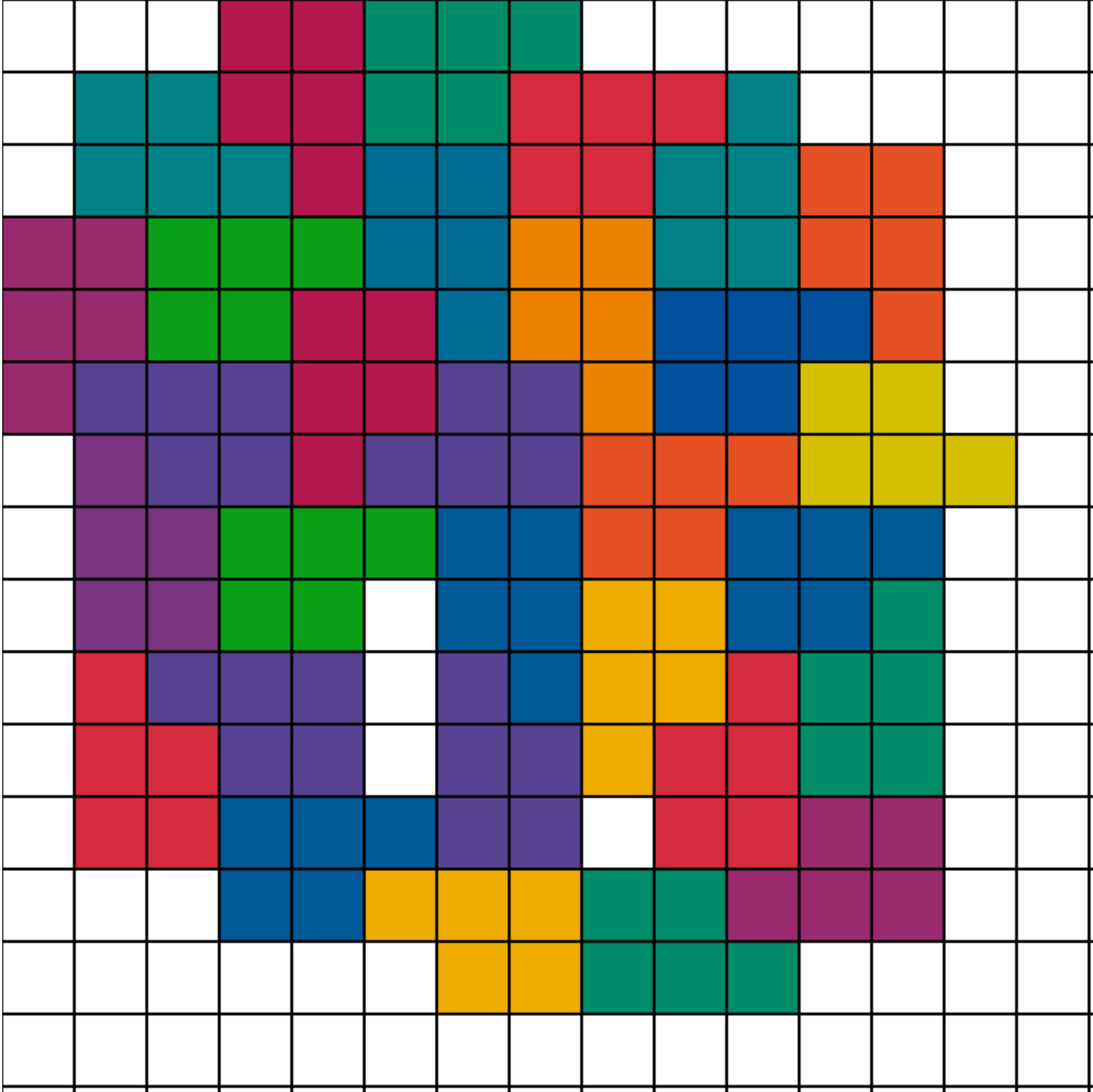}
\includegraphics[width=0.3\textwidth]{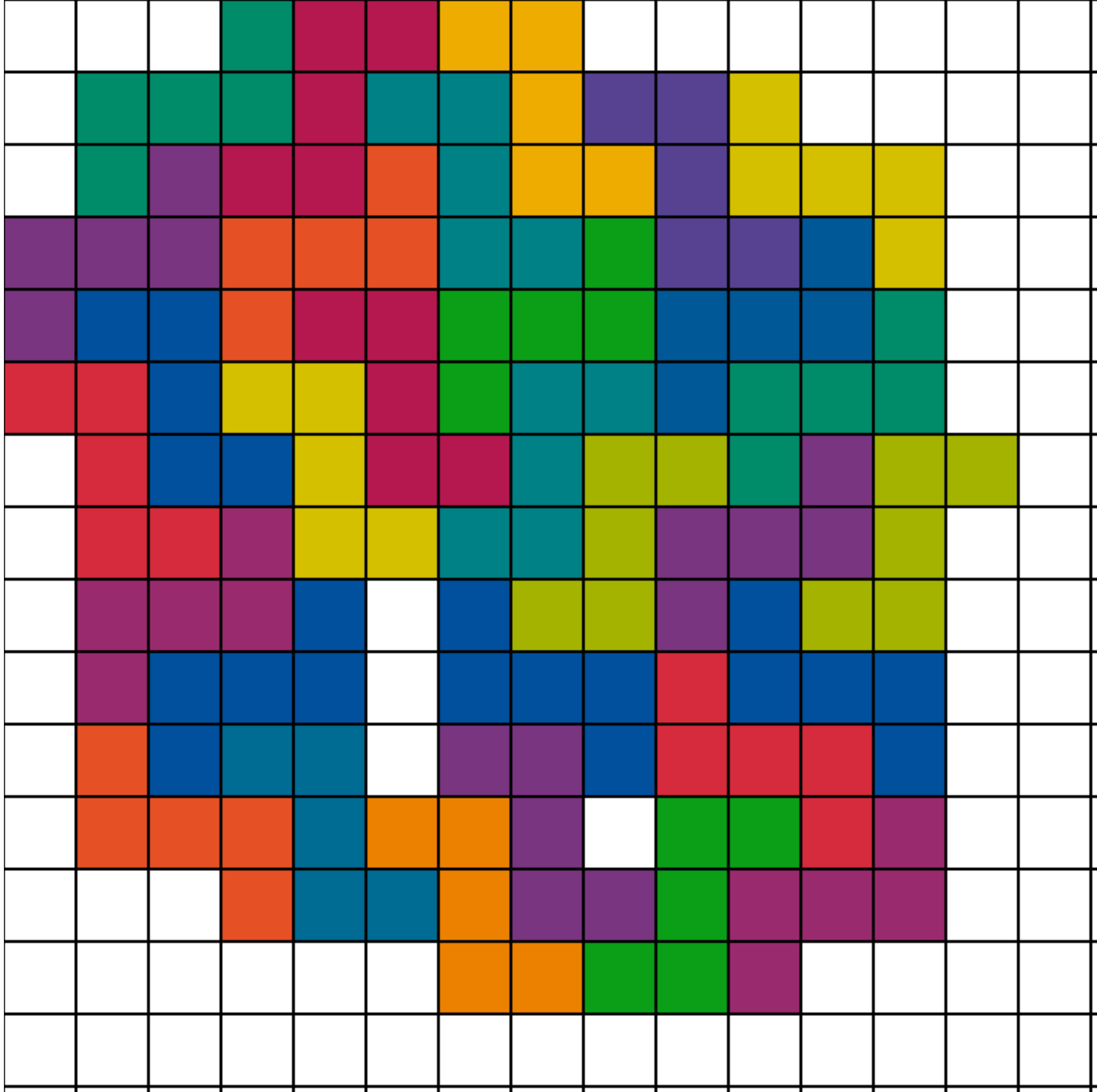}
\includegraphics[width=0.3\textwidth]{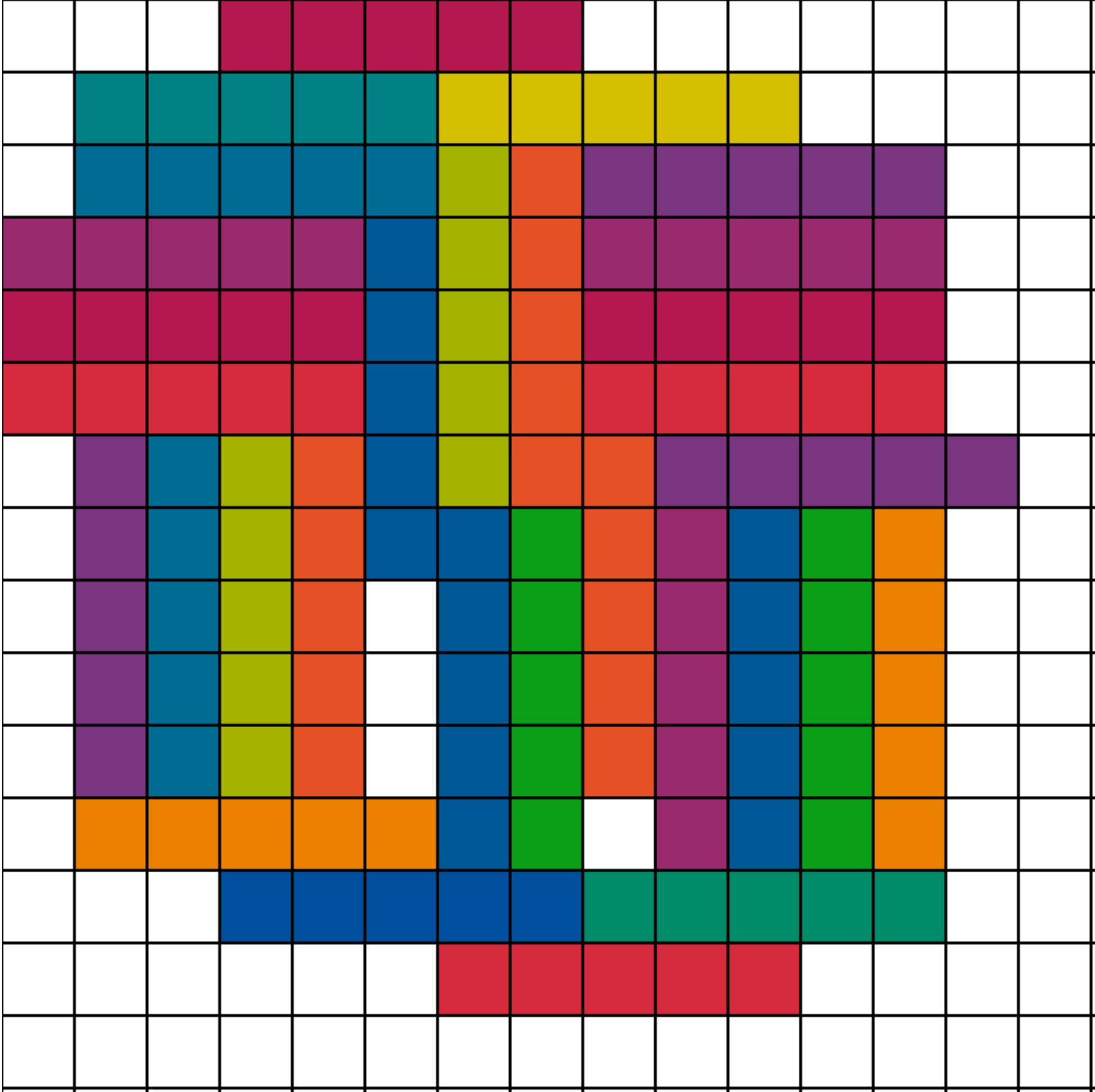}
\caption{Tiling patterns for I5P5Z5 improved from 200-omino to 150-omino}
\label{fig:I5-P5-Z5}
\end{figure} 

\begin{figure}[thb]\centering
\includegraphics[width=0.3\textwidth]{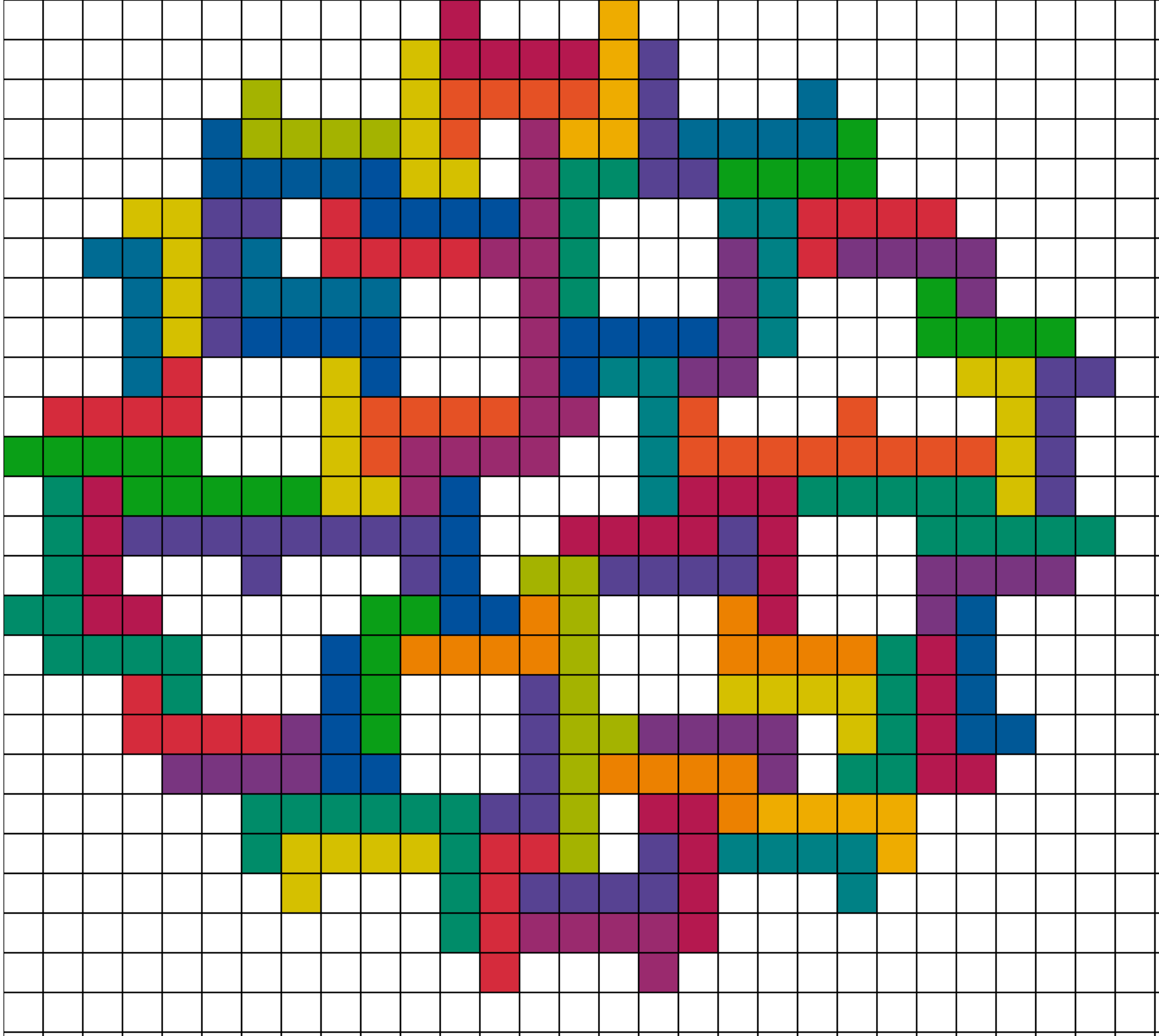}
\includegraphics[width=0.3\textwidth]{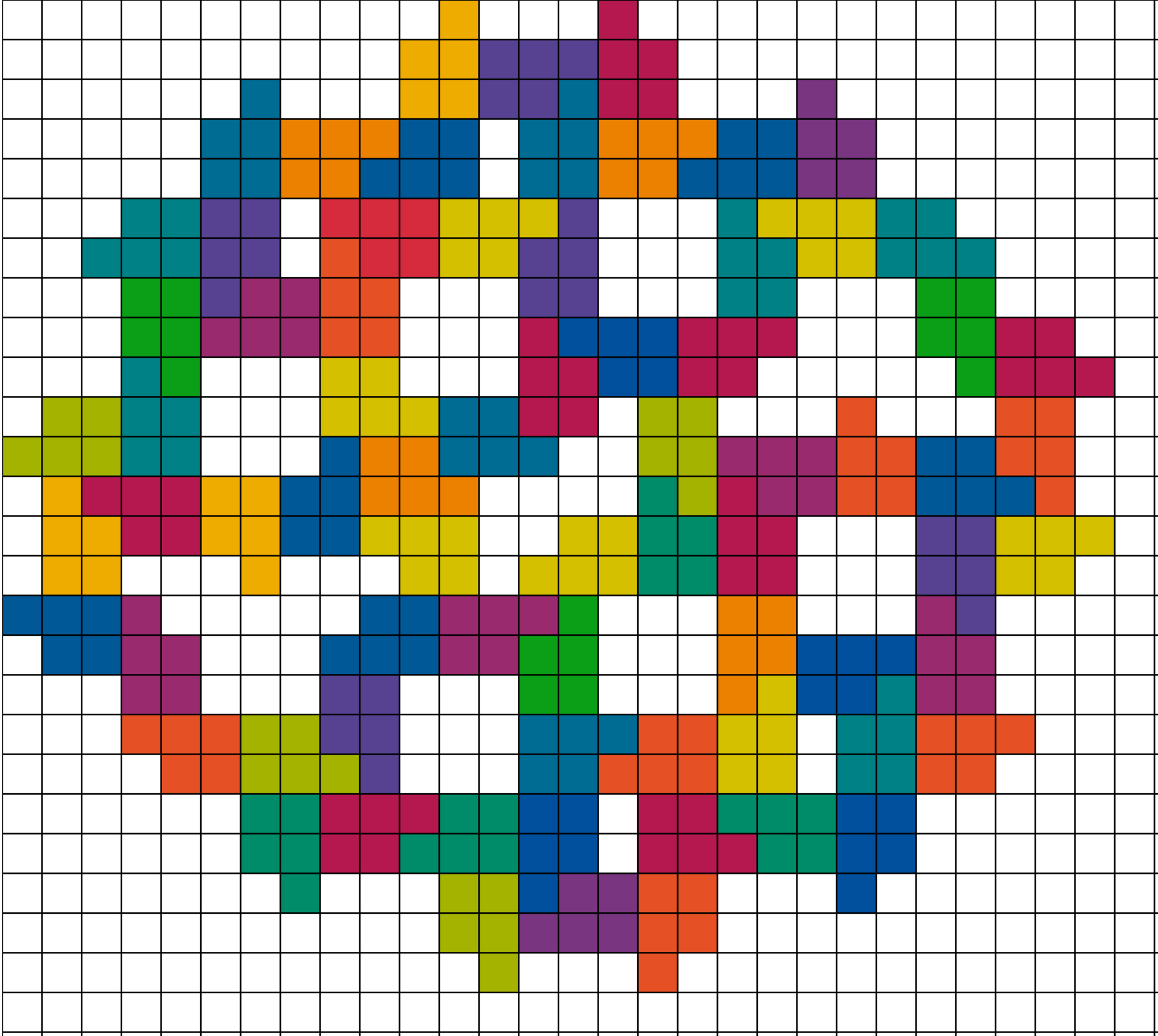}
\includegraphics[width=0.3\textwidth]{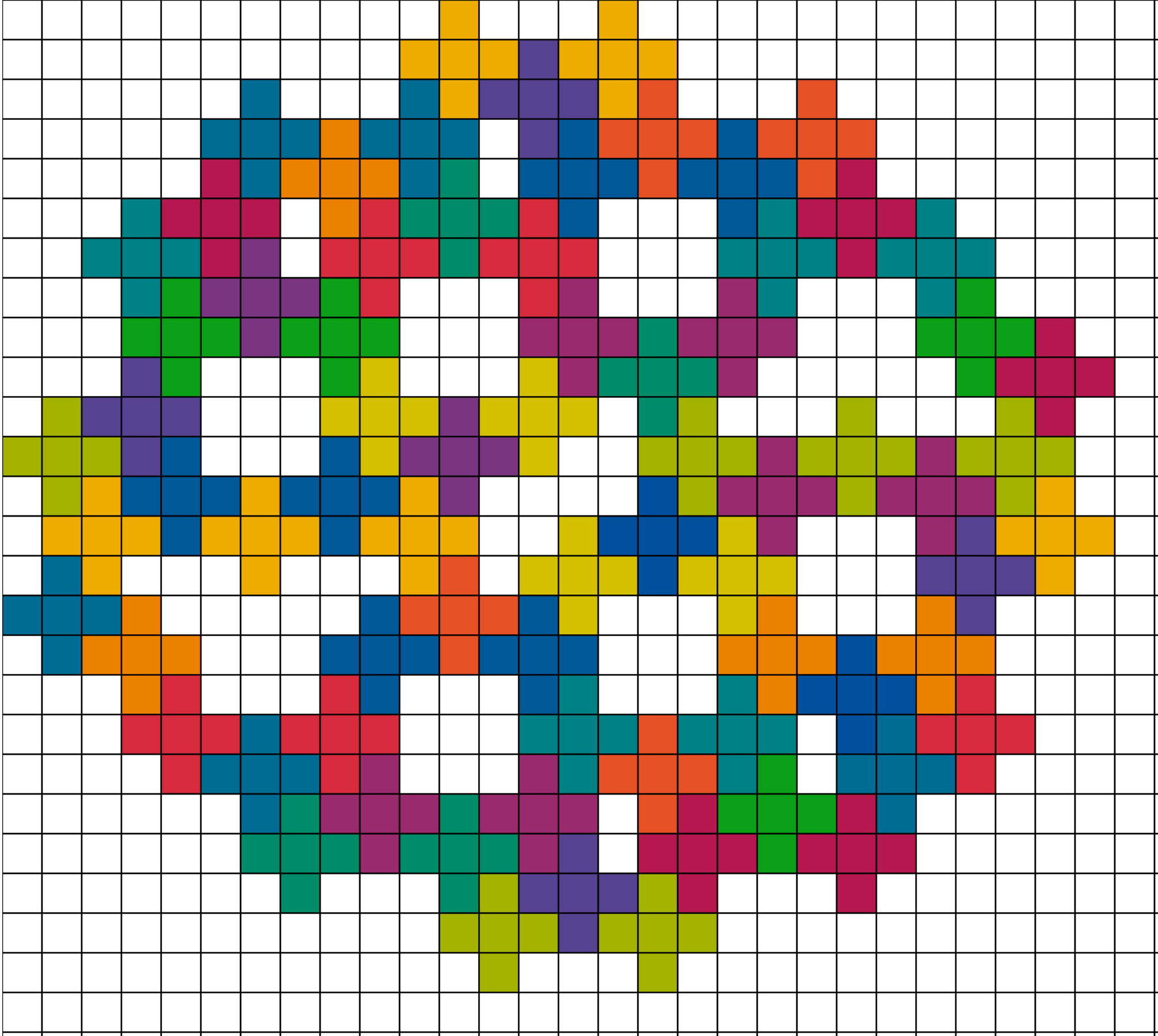}
\caption{Tiling patterns for L5P5X5 improved from 400-omino to 360-omino}
\label{fig:L5-P5-X5}
\end{figure}

\begin{figure}[thb]\centering
\includegraphics[width=0.3\textwidth]{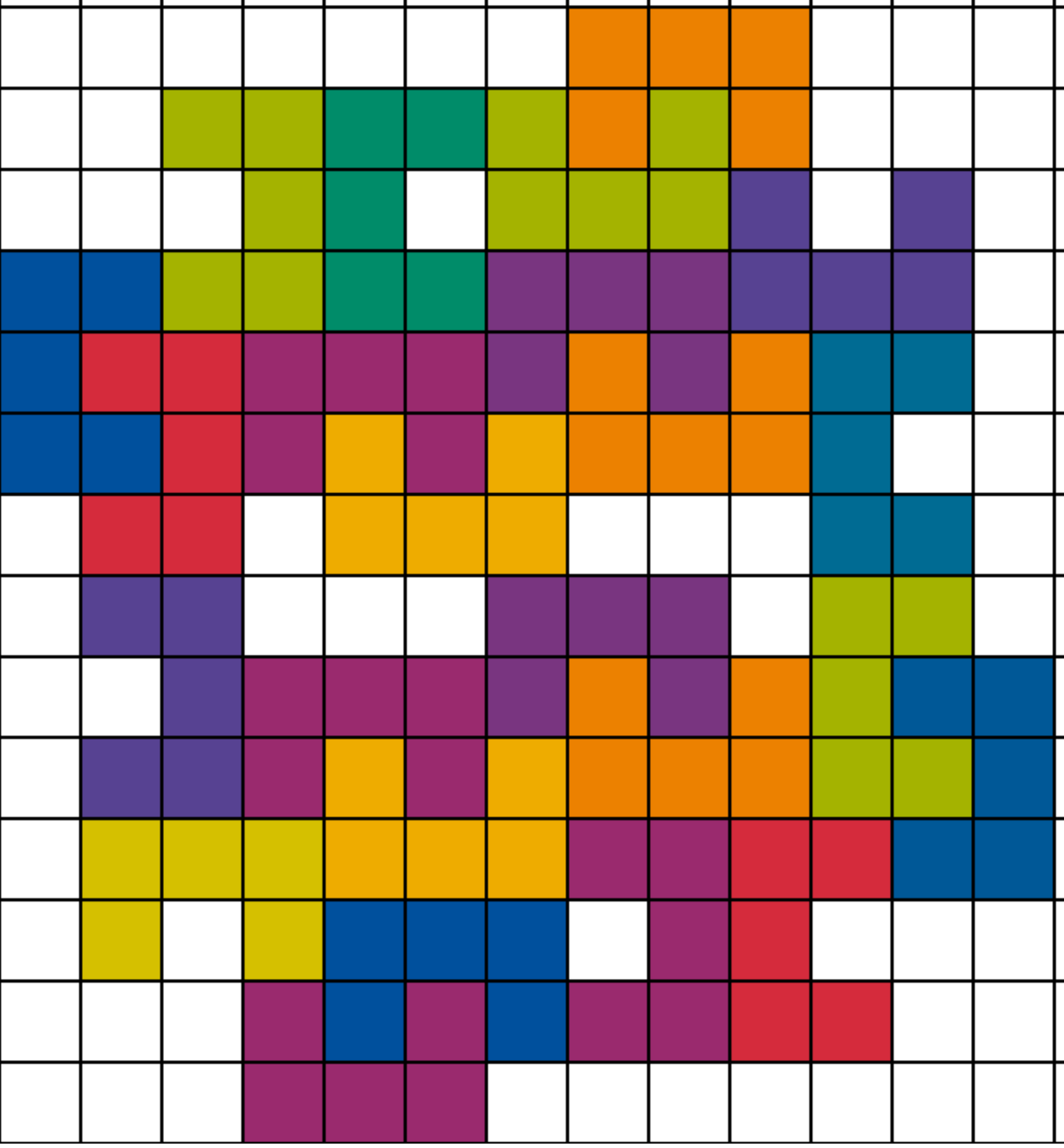}
\includegraphics[width=0.3\textwidth]{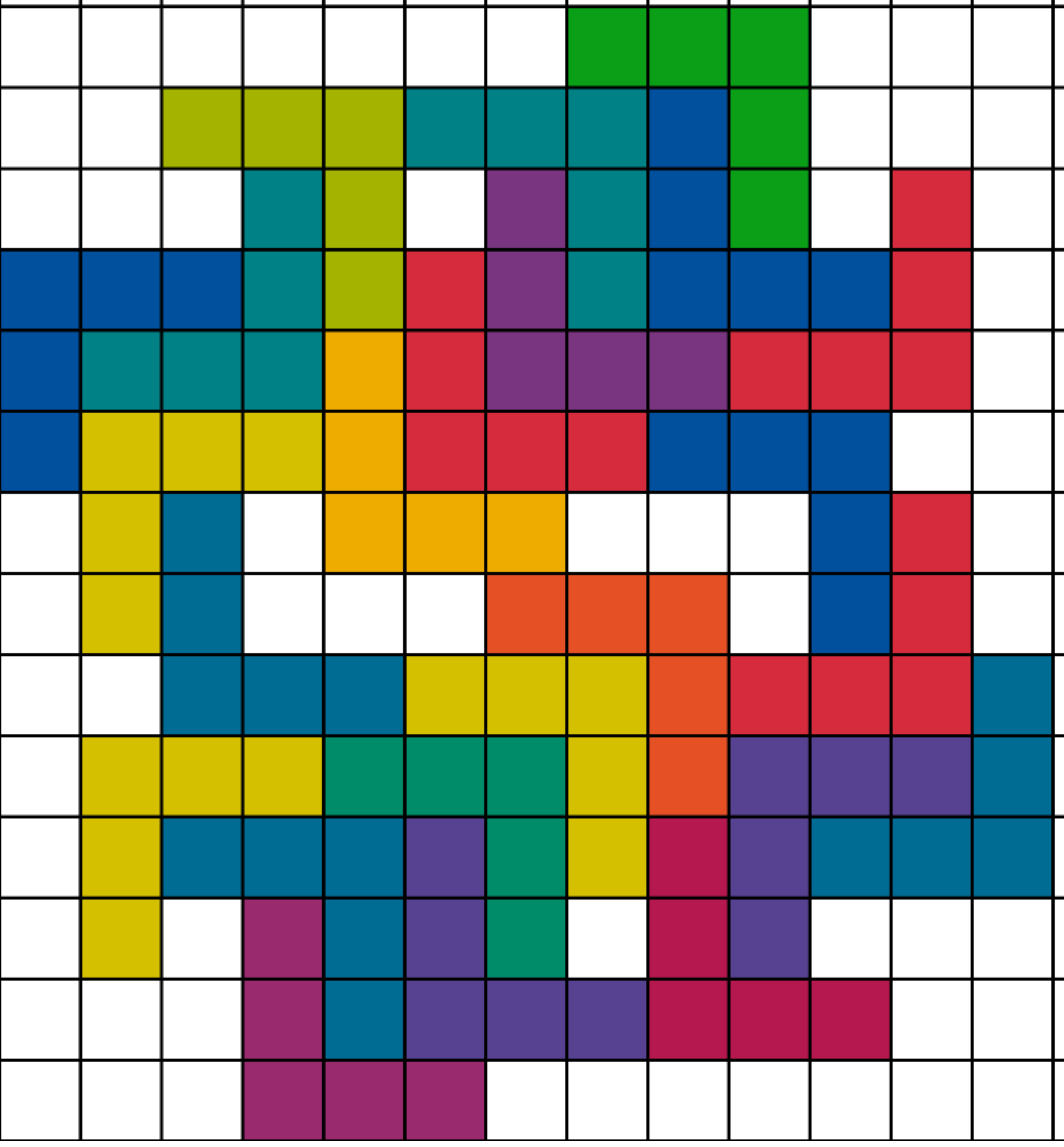}
\includegraphics[width=0.3\textwidth]{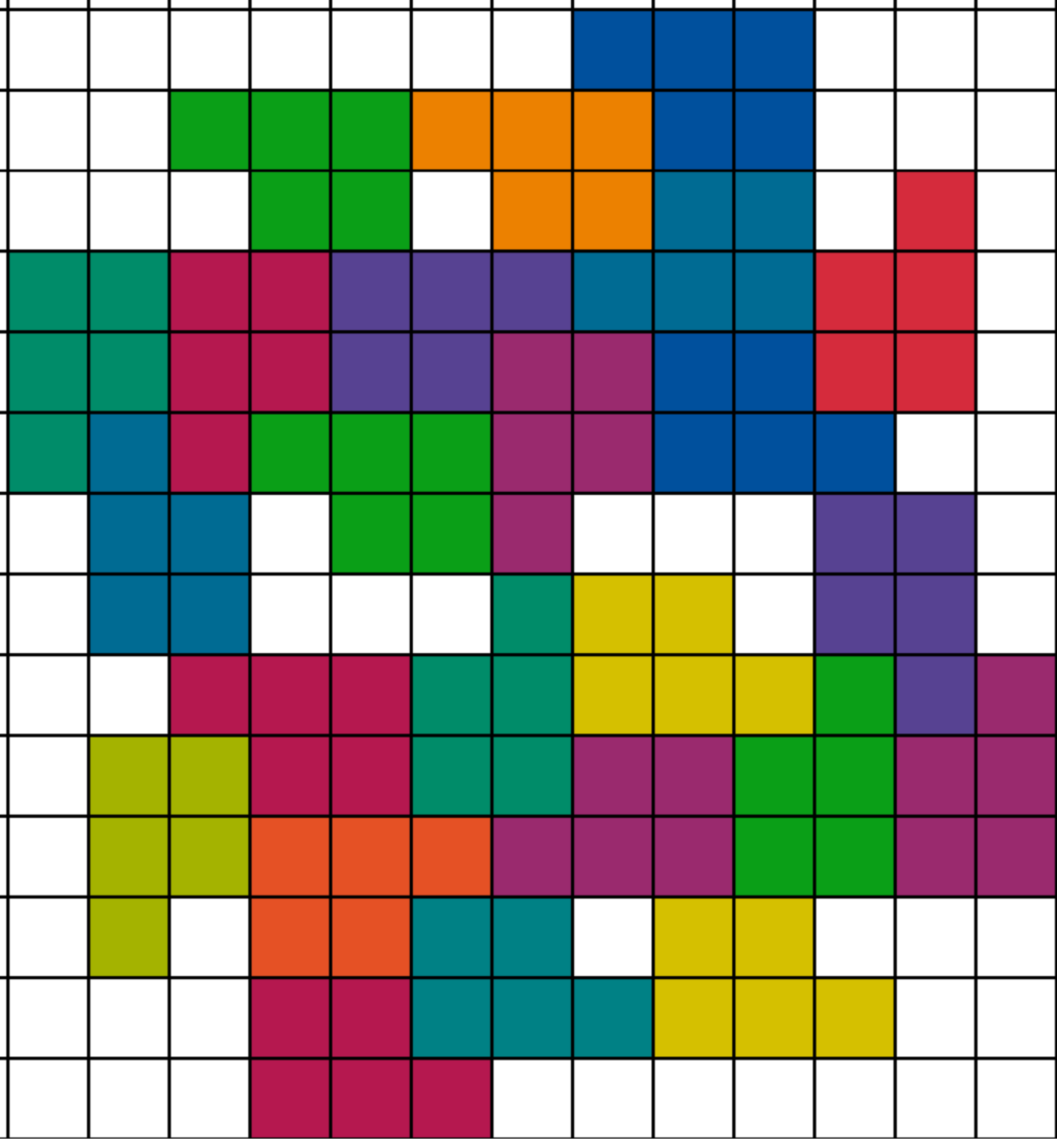}
\caption{Tiling patterns for V5U5P5 improved from 160-omino to 120-omino}
\label{fig:P5-U5-V5}
\end{figure} 

Two main differences exist between the formulations of the common shape puzzle and the packing puzzle in \cite{BHHM2021}.
The first one is that the goal shape is not provided in the common shape puzzle, whereas it is provided in the packing puzzle.
The second is that we must create a common (or congruent) shape 
using two different sets $\calS_1$ and $\calS_2$ of pieces in the common shape puzzle, 
whereas we have only one set of pieces in the packing puzzle.

To address the first point, we fix the bounding box of the goal shape.
We first fix the number of pieces (or $\msize{\calS_1}$ and $\msize{\calS_2}$),
and we attempt to create possible bounding boxes that contains these pieces.

In the packing puzzle, we can assume that each unit square of a goal shape is covered by exactly once by a piece.
However, in the common shape puzzle, each unit square of a bounding box is covered by either 0 or 2 pieces.
Moreover, when the square is covered by 2 pieces, these should be in $\calS_1$ and $\calS_2$.

We can modify the formulation of the packing puzzle in \cite{BHHM2021} to that for the common shape puzzle using these concepts.
Furthermore, it is straightforward to extend the problem from two sets $\calS_1$ and $\calS_2$ to 
three sets $\calS_1$, $\calS_2$, and $\calS_3$ (and more).
We investigated several cases that are available online, and achieved some improvements, as outlined in Appendix~\ref{sec:patterns}.

\section{Concluding Remarks}
We have considered the computational complexities of a generalized common shape puzzles, in which the goal shapes are not provided.
The puzzle is tractable when the number of pieces is a constant; however, 
it is strongly NP-complete even if the piece sets consist of small rectangles.
Moreover, if we are allowed to use the copies of the pieces repeatedly, the problem becomes undecidable.
It is possible to formulate the puzzle for some several different solvers in a natural form,
and we improved some known records for concrete instances using a SAT-based solver.
However, we have not yet succeeded in confirming that the results are the minimum solutions.
For example, we verified the pattern in \figurename~\ref{fig:5T=4Q} 
for each boundary box with a size of $i\times \floor{625/i}$ using $1\le i\le 25$ and confirmed that
there are no smaller patterns in these boundary boxes. 
However, this does not imply that the pattern in \figurename~\ref{fig:5T=4Q} is the least pattern.
Thus, efficient searching for the determination of the minimum solution remains open.
We have only considered the polyominoes in this study, and thus, the extension to general polygons is 
a natural topic for future work.

\bibliography{main}

\appendix

% \section{A complete example of a construction}
% \label{sec:const}
% \input{app-const}

% \clearpage

% \section{Improved patterns}
% \label{sec:patterns}
% \input{patterns}

\end{document}